\documentclass[a4paper, 
twocolumn,
aps,
floatfix,
superscriptaddress, accepted=2020-12-27]{quantumarticle}
\pdfoutput=1

\usepackage{amsmath,amssymb,amsthm}
\usepackage[margin=2cm]{geometry}
\usepackage{tikz}
\usepackage{array}
\usepackage{mathtools}
\usepackage{multirow}
\usepackage{hhline}
\usepackage{physics}
\usepackage{qcircuit}
\usepackage{setspace}
\usepackage{makecell}

\usetikzlibrary{backgrounds,shadows.blur,fit,decorations.pathreplacing,shapes}
\usetikzlibrary{shapes,arrows,chains,calc}

\newcommand{\Send}{S_{\text{end}}}
\newcommand{\Sstart}{S_{\text{start}}}
\newcommand{\Sgraph}{S_{\text{graph}}}

\newcommand{\Var}{\mathrm{Var}}
\newcommand{\diff}{\mathop{}\!\mathrm{d}}
\newcommand{\covs}[2]{\text{Cov}[#1, #2]}
\newcommand{\trans}{^\mathsf{T}}

\newtheorem{claim}{Claim}
\newtheorem{theorem}{Theorem}

\newcolumntype{C}[1]{>{\centering\let\newline\\\arraybackslash\hspace{0pt}}m{#1}}
\newcolumntype{R}[1]{>{\raggedleft\let\newline\\\arraybackslash\hspace{0pt}}m{#1}}

\makeatletter
\renewcommand*\env@matrix[1][*\c@MaxMatrixCols c]{%
  \hskip -\arraycolsep
  \let\@ifnextchar\new@ifnextchar
  \array{#1}}
\makeatother

\begin{document}

\singlespacing

\title{Efficient quantum measurement of Pauli operators \newline in the presence of finite sampling error
}

\author{Ophelia Crawford$^* $}
\author{Barnaby van Straaten$^* $}
\affiliation{Riverlane, Cambridge, UK}
\author{Daochen Wang$^* $}
\affiliation{Riverlane, Cambridge, UK}
\affiliation{Joint Center for Quantum Information and Computer Science, University of Maryland, College Park, USA}
\author{Thomas Parks}
\affiliation{Riverlane, Cambridge, UK}
\author{Earl Campbell}
\affiliation{Riverlane, Cambridge, UK}
\affiliation{Department of Physics and Astronomy, University of Sheffield, Sheffield, UK}
\author{Stephen Brierley}
\affiliation{Riverlane, Cambridge, UK}

\begin{abstract}
{Estimating the expectation value of an operator corresponding to an observable is a fundamental task in quantum computation. It is often impossible to obtain such estimates directly, as the computer is restricted to measuring in a fixed computational basis. One common solution splits the operator into a weighted sum of Pauli operators and measures each separately, at the cost of many measurements. An improved version collects mutually commuting Pauli operators together before measuring all operators within a collection simultaneously. The effectiveness of doing this depends on two factors. Firstly, we must understand the improvement offered by a given arrangement of Paulis in collections. In our work, we propose two natural metrics for quantifying this, operating under the assumption that measurements are distributed optimally among collections so as to minimise the overall finite sampling error. Motivated by the mathematical form of these metrics, we introduce S\scriptsize ORTED \normalsize I\scriptsize NSERTION\normalsize, a collecting strategy that exploits the weighting of each Pauli operator in the overall sum. Secondly, to measure all Pauli operators within a collection simultaneously, a circuit is required to rotate them to the computational basis. In our work, we present two efficient circuit constructions that suitably rotate any collection of $\boldsymbol{k}$ independent commuting $\boldsymbol{n}$-qubit Pauli operators using at most  $\boldsymbol{kn-k(k+1)/2}$ and $\boldsymbol{O(kn/\log k)}$ two-qubit gates respectively. Our methods are numerically illustrated in the context of the Variational Quantum Eigensolver, where the operators in question are molecular Hamiltonians. As measured by our metrics, S\scriptsize ORTED \normalsize I\scriptsize NSERTION \normalsize outperforms four conventional greedy colouring algorithms that seek the minimum number of collections.}
\end{abstract}

\maketitle
\def\thefootnote{*}\footnotetext{These authors contributed equally to this work.}

\section{Introduction}

Estimating the expectation value of an operator corresponding to an observable on a quantum state is a fundamental task in quantum mechanical experiments. Expectation estimation of a Hamiltonian features prominently as the quantum sub-routine of the Variational Quantum Eigensolver (VQE) algorithm~\cite{peruzzo2014variational}, which has emerged as a leading candidate for exhibiting quantum advantage in the Noisy Intermediate Scale Quantum era~\cite{preskill2018quantum}. VQE is a hybrid quantum-classical algorithm designed to find the ground state~\cite{peruzzo2014variational, mcclean2016theory, omalley2016scalable, kandala2017_hardware, mcardle_quantum_2018, ryabinkin_constrained_2018, romero2018strategies, wang2018generalised}, or energy spectrum~\cite{mcclean2017hybrid,santagati2018witnessing,colless2018computation,heya_subspace_2019,jones_variational_2019,Higgott2019variationalquantum}, of a physical or chemical system. Often there is no natural way to obtain the expectation of an operator directly and some indirect protocol is required. In particular, this is true of current quantum computers, which can only obtain measurements in the computational basis defined, by convention, as eigenstates of the Pauli-$Z$ operator on each qubit.

One naive method of proceeding, therefore, is to decompose the operator into a weighted sum of Pauli operators (or Paulis) and then measure each separately. The paper that introduced VQE~\cite{peruzzo2014variational} proposed measuring the Hamiltonian in this way. However, this can be inefficient. For example, a second-quantised chemical Hamiltonian on $n$ qubits decomposes into a very large number of Paulis that scales as $n^{4}$. An improved method, therefore, is to arrange the Paulis in commuting collections. All Paulis in a collection can be measured at the same time, as any set of commuting Paulis can be simultaneously diagonalised by a single unitary. There are typically many arrangements of Paulis in mutually commuting collections, and the reduction in the number of measurements required will depend on the arrangement.

In the context of VQE, \citeauthor{mcclean2016theory}~\cite{mcclean2016theory} proposed this improved protocol and argued using a toy example that the arrangement with the fewest collections might not result in the fewest measurements as splitting a single commuting collection into two might sometimes offer an improvement. However, Ref.~\cite{mcclean2016theory} did not show how to obtain commuting collections nor how to construct the unitary rotation, $U$, that enables simultaneous measurement.

Recently, a series of papers~\cite{izmaylov2019_minclique, mosca2019_pauli_partioning, izmaylov2019_general_commute, babbush2019_chem_commute, gokhale2019_commute, zhao_2019, gokhale_n3}\footnote{We mention that both the collecting and rotation circuit synthesis problems have been addressed on an ad-hoc basis by experimentalists since~\citeauthor{kandala2017_hardware}~\cite{kandala2017_hardware}. We refer readers to Table~2 of Ref.~\cite{gokhale2019_commute} for a good summary. We also mention that recent Ref.~\cite{2019izmaylov_anticommute} and the less recent Ref.~\cite{izmaylov2018_revising} allow for collecting of non-commuting Paulis. We found it interesting that such approaches are feasible, but found it hard to compare them to our work on a like-for-like basis.} have appeared that make progress on both the collecting strategy and rotation circuit construction problems. Our paper is in the same arena and addresses \textit{both} problems.

First, we define two natural metrics, $R$ and $\hat{R}$, that quantify the performance of any given arrangement. $R$ and $\hat{R}$ measure the ratio of the number of measurements required in the uncollected case to the collected case to attain a fixed level of accuracy. The key feature of these two metrics is that they assume measurements are distributed optimally between the collections to minimise the finite sampling error, following Refs~\cite{wecker2015progress, Rubin_2018,romero2018strategies}. The difference between them is that $R$ is state-dependent but $\hat{R}$ is designed to approximate $\mathbb{E}[R]$ over the uniform spherical measure. Therefore, $R$ is more suitable for use given knowledge of the underlying quantum state, while $\hat{R}$ is more suitable otherwise.

With $R$ and $\hat{R}$ defined, we prove a first result that breaking a single commuting collection into two is never advantageous, as the number of measurements required to obtain an expectation estimate to a given accuracy is never reduced. This result already contradicts the conclusion of the aforementioned toy example used by~\citeauthor{mcclean2016theory}~\cite{mcclean2016theory}, and analysed in full in Ref.~\cite[Sec.~10.1]{gokhale2019_commute}, that breaking a collection can be advantageous. The reason for the discrepancy is that we distribute measurements optimally among the collections so as to minimise the overall error due to finite sampling for a given number of measurements. In contrast, in the previous works, measurements are distributed uniformly between the collections.

We further propose an intrinsically new collecting strategy, which we call \textsc{Sorted Insertion}, that is designed to maximise $R$ and $\hat{R}$. Unlike all strategies used previously that seek the minimum number of collections~\cite{izmaylov2019_minclique, mosca2019_pauli_partioning,gokhale2019_commute,zhao_2019}, \textsc{Sorted Insertion} seeks to maximise $\hat{R}$ by explicitly exploiting the coefficients of the Paulis in its assignment of collections. It is important to stress that maximising $\hat{R}$ is not the same as minimising the number of collections, as we show by a toy example. Perhaps counter-intuitively, this does not contradict the above conclusion, which showed only that it is never better to break a single collection into two.

When considering the rotation circuit construction problem, we contribute two new methods for constructing Clifford rotation circuits that enable simultaneous measurement of a collection containing arbitrary commuting Paulis. Like Ref.~\cite{gokhale2019_commute}, we approach the problem via the stabiliser formalism but further consider the case where the number of independent operators, $k$, in a collection can be less than the number of qubits, $n$. We show that the number of two-qubit gates in the rotation circuit, $U$, can be reduced in a way that scales with $k$. This is important because it is atypical for actual collections to have exactly $k=n$ independent Paulis and reducing the number of two-qubit gates is important, especially in the near-term~\cite{blais_quantum-information_2007,obrien2010_photon,martinis2014_google,ballance_high-fidelity_2016,allen_optimal_2017,monroe2017_pnas_compare,wendin2017_superconductor_fidelity,reagor2018_rigetti,schafer_fast_2018,webb_resilient_2018,lukin2018_rydberg,simmons2019_electron_spin,huang2019_quantum_dot,blumel_power-optimal_2019,he_two-qubit_2019}. As far as we are aware, ours is the first paper to consider the $k<n$ case explicitly. Furthermore, we allow classical post-processing, which can save quantum resources.

More specifically, we introduce constructions ``\textsc{CZ}'' and ``\textsc{CNOT}''. The \textsc{CZ}-construction builds on work by \citeauthor*{vandennest2004_graph_states}~\cite{vandennest2004_graph_states} in the graph-state literature to yield $U$ with a number of two-qubit gates at most
\begin{equation}
    u_{\text{cz}}(k,n) = kn-k(k+1)/2.
\end{equation}
The \textsc{CNOT}-construction builds on our \textsc{CZ}-construction, and work by \citeauthor{aaronson_gottesman_simulation}~\cite{aaronson_gottesman_simulation} and~\citeauthor*{patel_markov_hayes_cnot}~\cite{patel_markov_hayes_cnot} to yield $U$ with a number of two-qubit gates at most
\begin{equation}
    u_{\text{cnot}}(k,n) = O(kn/\log k).
\end{equation}
We stress that $u_{\text{cnot}}$ and $u_{\text{cz}}$ are worst-case upper bounds. In practice, numerical simulations are needed to determine whether the \textsc{CZ}- or \textsc{CNOT}-construction is actually more efficient. We note that, in the case of $k=n$, our constructions have two-qubit gate-counts scaling no worse than the previous best of $O(n^{2})$~\cite{gokhale2019_commute}. Other works, such as Ref.~\cite[Appendix~A]{mosca2019_pauli_partioning}, do not provide gate-counts, or demonstrate a scaling that is worse than $O(n^{2})$~\cite[Appendix~B]{izmaylov2019_general_commute}, or present a method that only works for the second-quantised fermionic Hamiltonian~\cite{babbush2019_chem_commute, gokhale_n3}. 

We end our paper with a series of numerical results on molecular Hamiltonians that serve to illustrate and validate our theoretical work. In doing so, we first quantify the performance of \textsc{Sorted Insertion} using the metric $\hat{R}$ for molecules ranging in size from hydrogen H$_{2}$, which requires two qubits, to hydrogen selenide H$_{2}$Se, which requires 38, finding that it leads to a 10 to 60 fold improvement in the number of measurements required. Note that we are defining a single measurement to consist of a measurement of all qubits, and so the number of measurements equals the number of ansatz state preparations. 

We further present results of using \textsc{Sorted Insertion} alongside our \textsc{CZ}-construction on molecules requiring up to 38 qubits to calculate the actual number of two-qubit gates required for real molecular systems. Our numerical results show that the typical number of two-qubit gates is fewer than the worst-case $u_{\text{cz}}(k,n)$ by a factor of 3.5. 

Lastly, we present data showing \textsc{Sorted Insertion} outperforming the four conventional greedy colouring algorithms we tested, as measured by the metric $\hat{R}$. Our data strongly \textit{challenges} the assumption that minimising the number of collections is best, as arrangements with the smallest number of collections do not typically perform the best with respect to $\hat{R}$.

We mention that an entirely different measurement strategy has been proposed recently~\cite{huang_2020} that is based on so-called ``classical shadows''. Ref.~\cite[Illustrative Example Applications]{huang_2020} explicitly mentions that their method is unsuitable for estimating expectation values of global observables. Even for a Hamiltonian consisting of a single global Pauli operator $P \coloneqq P_1\otimes P_2\otimes\dots\otimes P_n$, with $P_i\in \{X,Y,Z\}$, their number of classical shadows must scale as $2^{\Omega(n)}$. In contrast, our method, as with the original VQE proposal, scales independently of $n$ and would simply measure each $P_i$ and multiply the outcomes.

Follow-up Ref.~\cite{hadfield_2020} sought to ameliorate this problem by using so-called ``locally-biased classical shadows''. However, as they mention in \cite[Remark 1]{hadfield_2020}, the \emph{analytical} upper bound on their estimator can still scale as $2^{\Omega(n)}$ for the same $P$ above. Separately, Ref.~\cite{hadfield_2020} did not compare their techniques to ours, rather they observe that we require deeper circuits. However, the additional depth our method requires is modest \emph{relative} to that of ansatz state preparation.

There have also been interesting developments~\cite{kubler_2020,arrasmith_2020} on measurement reduction in VQE focusing on designing classical optimisers that intelligently choose how many measurements to use at each iteration -- for example, fewer measurements may be used near the start versus near the end. These methods complement ours and can be combined with ours to reduce further the overall number of measurements.

\section{Collecting strategies}\label{sec:collections}
In this section, we develop a method for collecting Pauli operators designed to minimise the number of measurements required to obtain an estimate of the expectation value of an operator to a given level of accuracy.

We have an operator, $O$, of the form
\begin{equation}\label{eq:collecting_strat_collecting}
    O  = \sum_{i=1}^{N} O_{i} = \sum_{i=1}^{N} \sum_{j=1}^{m_{i}} a_{ij} P_{ij},
\end{equation}
where $N$ is the number of collections of mutually commuting operators, $m_{i}$ is the number of operators in collection $i$, $P_{ij}$ is the $j$th Pauli operator in the $i$th collection and $a_{ij}\in \mathbb{R}$ is its coefficient. 

We assume that the $P_{ij}$s are distinct for different $(i,j)$, i.e., we do not collect a single Pauli operator appearing in $O$ into more than one collection by splitting its coefficient. This is discussed further at the end of this section.

Given $\epsilon$, let $M_{u}$ and $M_{g}$ be the minimum number of measurements required to attain an accuracy $\epsilon$ in the uncollected and collected (as per Eq.~\ref{eq:collecting_strat_collecting}) cases respectively. Finding $M_{u}$ is a special case of finding $M_{g}$. To find $M_{g}$, we can solve the constrained optimisation problem that asks how a given number of measurements should be distributed in order to maximise accuracy. The solution to such a problem gives the optimal measurement strategy for a given operator and state. Using Lagrange multipliers~\cite{wecker2015progress, Rubin_2018,romero2018strategies}, we find that the optimal number $n_i$ of measurements of collection $i$ is
\begin{equation}
    n_i\coloneqq\frac{1}{\epsilon^2}\sqrt{\Var[O_i]}\ \sum_{j=1}^N\sqrt{\Var[O_j]}.
\end{equation}

Therefore, we have
\begin{equation}
M_g = \sum_{i=1}^N n_i = \left (\frac{1}{\epsilon}\sum_{i=1}^{N} \sqrt{\Var[O_{i}]}  \right )^{2},
\end{equation}
where 
\begin{align}
    \Var[O_{i}] &= \covs{O_{i}}{O_{i}} \\
    &= \expval{O_{i}^{2}} - \expval{O_{i}}^{2}.
\end{align}
Since $M_u$ is just $M_g$ evaluated with every operator in its own collection, we have
\begin{equation}
M_{u} = \left ( \frac{1}{\epsilon} \sum_{i=1}^{N} \sum_{j=1}^{m_{i}} |a_{ij}|\sqrt{\Var[P_{ij}]} \right)^{2},
\end{equation}
where 
\begin{equation}
    \Var[P_{ij}] = 1- \expval{P_{ij}}^{2}.
\end{equation}
The ratio $R$, defined as
\begin{equation}\label{eq:rratio}
    R \coloneqq \frac{M_u}{M_g} = \left( \frac{\sum_{i=1}^{N} \sum_{j=1}^{m_{i}} |a_{ij}|\sqrt{\Var[P_{ij}]}}{\sum_{i=1}^{N} \sqrt{\Var[O_{i}]}}\right)^{2},
\end{equation}
therefore acts as a natural metric for the performance of a particular arrangement of Paulis under the assumption that measurements are distributed optimally. The larger the value of $R$, the greater the saving that assembling operators into collections provides. In Theorem~\ref{theorem:breakingbad} below, we show that combining two collections into one can only improve $R$.

\begin{theorem}\label{theorem:breakingbad}
    Consider two collections $\mathcal{A}$ and $\mathcal{B}$ of mutually commuting Paulis, where each Pauli is in at most one collection. Suppose that it is possible to combine $\mathcal{A}$ and $\mathcal{B}$ into a single commuting collection $\mathcal{C}\coloneqq \mathcal{A} \cup \mathcal{B}$. Let $R(\{\mathcal{A},\mathcal{B}\})$ and $R(\{\mathcal{C}\})$ denote the $R$ metric, as defined in Eq.~\eqref{eq:rratio}, for the two collections separated and combined respectively. Then
    \begin{equation}\label{eq:breaking_bad}
        R(\{\mathcal{A},\mathcal{B}\}) \leq R(\{\mathcal{C}\}) \ \ \text{ for all } \ \ket{\psi}.
    \end{equation}
\end{theorem}
\begin{proof}
We write the operators associated with collections $\mathcal{A}$ and $\mathcal{B}$ as
\begin{align}
    O_{\mathcal{A}} &= \sum_{j=1}^{k} a_{j} P_{j},\\
    O_{\mathcal{B}} &= \sum_{j=k+1}^{k+l} a_{j} P_{j},
\end{align}
for some $k,l$ respectively, where the indexing of the sum in $O_{\mathcal{B}}$ is for later convenience. The operator associated with $\mathcal{C}$ can therefore be written as
\begin{equation}
    O_{\mathcal{C}} = \sum_{j=1}^{k+l} a_{j} P_{j}.
\end{equation}

We define the $(k+l)\times(k+l)$ covariance matrix $C$, which is symmetric and positive semi-definite, associated with collection $\mathcal{C}$ by its entries
\begin{equation}
    C_{ij} \coloneqq \covs{P_i}{P_j},
\end{equation}
with
\begin{equation}
    \covs{P_i}{P_j}\coloneqq \expval{P_i P_j}_{\psi}- \expval{P_i}_{\psi}\expval{P_j}_{\psi}.
\end{equation}

We write $\mathbf{a}$ for the vector of length $k+l$ whose first $k$ entries are given by $\{a_{j}\}_{j=1}^{k}$ and the remaining $l$ entries are zero. Similarly, we write $\mathbf{b}$ for the vector of length $k+l$ whose first $k$ entries are zero and the remaining $l$ entries are given by $\{a_{j}\}_{j=k+1}^{k+l}$.

Since the numerator of $R$ is constant for different arrangements, to prove our claim it suffices for us to show that the denominator of $R$ is not smaller in the case of $\{\mathcal{A}, \mathcal{B}\}$ than in the case of $\{\mathcal{C}\}$, i.e., that
\begin{equation}
    \sqrt{\Var[O_{\mathcal{A}}]} + \sqrt{\Var[O_{\mathcal{B}}]} \geq \sqrt{\Var[O_{\mathcal{C}}]}.
\end{equation}
With the above notation, this is equivalent to
\begin{equation}
\sqrt{\mathbf{a}\trans C \mathbf{a}} + \sqrt{\mathbf{b}\trans C\mathbf{b}} \geq \sqrt{(\mathbf{a+b})\trans C  (\mathbf{a+b})}.
\end{equation}
We find
\begin{equation}
\begin{aligned}
    &\sqrt{\mathbf{a}\trans C \mathbf{a}} + \sqrt{\mathbf{b}\trans C\mathbf{b}} \\
    &= \sqrt{\left(\sqrt{\mathbf{a}\trans C\mathbf{a}} + \sqrt{\mathbf{b}\trans C\mathbf{b}}\right)^2} \\
    &= \sqrt{\mathbf{a}\trans C\mathbf{a} + \mathbf{b}\trans C\mathbf{b} +2 \sqrt{(\mathbf{a}\trans C\mathbf{a}) (\mathbf{b}\trans C\mathbf{b}})} \\
    &\geq \sqrt{\mathbf{a}\trans C\mathbf{a} + \mathbf{b}\trans C\mathbf{b} +2(\mathbf{a}\trans  C \mathbf{b}})\\
    &= \sqrt{(\mathbf{a+b})\trans  C  (\mathbf{a+b})},
\end{aligned}
\end{equation}
where the fourth line follows due to the Cauchy-Schwarz inequality on a semi-inner product defined by $\langle{\mathbf{a}},{\mathbf{b}}\rangle \coloneqq \mathbf{a} \trans C \mathbf{b}$~\cite[Example 1.1]{conway1994course} and the fifth line follows as $C$ is symmetric. This establishes our claim. 

The inequality of Eq.~\ref{eq:breaking_bad} holds with equality if and only if there exist $0\neq \alpha, \beta \in$  $\mathbb{C}$, such that $\langle{\alpha \mathbf{a} + \beta \mathbf{b}},{\alpha \mathbf{a} + \beta \mathbf{b}}\rangle = 0$ \cite[Example 1.4]{conway1994course}. 
\end{proof}

Theorem~\ref{theorem:breakingbad} shows that it is impossible to mitigate covariances by splitting collections under the optimal measurement strategy. This is in contrast to Refs~\cite{mcclean2016theory, gokhale2019_commute}, which showed that it is possible using a sub-optimal measurement strategy. In Appendix~\ref{app:groupcomb}, we re-do precisely their example using the optimal measurement strategy. Note that Theorem~\ref{theorem:breakingbad} simply implies that it is never better to break a single collection into two, and not that the minimum number of collections is necessarily the best; indeed, we provide a counter-example below. 

If all of the variances in $R$ are replaced by their expectation values over the uniform spherical measure (see Ref.~\cite[Ch.~7]{watrous_2018}), we obtain another metric, $\hat{R}$, given by
\begin{equation}
    \label{eq:rhatratio}
    \hat{R} \coloneqq \left(\frac{\sum_{i=1}^{N}\sum_{j=1}^{m_i} |a_{ij}|}{\sum_{i=1}^{N} \sqrt{\sum_{j=1}^{m_i}|a_{ij}|^2}
    }\right)^2.
\end{equation}
The derivation of $\hat{R}$ is given in Appendix~\ref{app:rhatder}. The same proof as in Theorem \ref{theorem:breakingbad} can be used to show that breaking a collection into two is never better when measured by $\hat{R}$; the only difference is the covariance matrix must be replaced by its expectation over the uniform spherical measure. 

$\hat{R}$ is a particularly useful metric because it approximates $R$, but can be calculated analytically. When ignorant of the quantum state $\ket{\psi}$, a good collecting strategy should aim to maximise $\hat{R}$. As mentioned above, this is emphatically not the same as minimising the number of collections. For example, let us consider the operator $O$, on two qubits, defined as
\begin{equation}
    O = 4 \, X_{1} + 4 \, X_{2} + Z_{2} + Z_{1}X_{2}.
\end{equation}

The operators cannot be assembled into a single commuting collection but there is a unique arrangement, $\mathcal{G}_{1}$, consisting of two collections, given by
\begin{equation}
    \mathcal{G}_{1} \coloneqq \Big\{ \{4 \, X_{1}, \ Z_{2}\},\{4 \, X_{2}, \ Z_{1}X_{2}\} \Big\},
\end{equation}
which has $\hat{R} = 1.47$ to two decimal places. $\mathcal{G}_{1}$ performs worse than an arrangement $\mathcal{G}_{2}$, with three collections, given by
\begin{equation}
    \mathcal{G}_{2} \coloneqq \Big\{ \{4 \, X_{1}, 4 \, X_{2}\},\{Z_{2}\}, \{Z_{1}X_{2}\} \Big\},
\end{equation}
which has $\hat{R} = 1.71$ to two decimal places.

Reflecting upon such examples and the nature of the square root function in the denominator of $\hat{R}$ leads us to propose a collecting algorithm that prioritises collecting Paulis with large coefficients.

This algorithm, which we name \textsc{Sorted Insertion}, can be described as follows. Given an operator
\begin{equation}
    O = \sum_{i=1}^{t} a_{j}P_{j},
\end{equation}
where $t$ is the total number of Pauli operators, the entire set $\{(a_{j}, P_{j})\}_{j=1}^{t}$ is sorted by the absolute value of coefficients $a_j$, so that $|a_{1}|\geq \dots \geq |a_{t}|$. Then, in the order $i=1,\dots t$, it is checked whether $P_{i}$ commutes with all elements in an existing collection. If it does, it is added to that collection. If not, a new collection is created and $P_{i}$ is inserted there. The collections are checked in order of their creation\footnote{We note that other collection orderings could also be considered, for example, ordering the collections by their values of $\sum_{j=1}^{m_{i}} |a_{ij}|^{2}$, where the $a_{ij}$ are the coefficients of the Paulis contained in collection $i$ at the point of ordering.}. The collections formed are tracked and outputted at the end once the final $P_{t}$ has been inserted.

\textsc{Sorted Insertion} has worst-case complexity $O(nt^{2})$, where we recall that $n$ is the number of qubits. In contrast, greedy colouring algorithms, as implemented in Ref.~\cite{izmaylov2019_minclique}, require pre-generating the \textit{full} commutation graph which has complexity $\Theta(nt^{2})$ even in the best case. The colouring algorithms then run on this graph adding further complexity -- see Table~\ref{tab:complexity}. Therefore, \textsc{Sorted Insertion} is at least as efficient as these greedy colouring algorithms.

\begin{table}[ht]
\centering
\begin{tabular}{@{}lc@{}}
\toprule
\\[1pt]
Colouring Algorithm         & Time Complexity
\\[1pt]
\hhline{--}
\\[1pt]
Largest First              & $O(t^2)$       
\\[1pt]
Connected Sequential d.f.s.& $O(t^2)$        
\\[1pt]
DSATUR                     & $O(t^2 \log t)$ 
\\[1pt]
Independent Set            & $O(t^3)$        
\\[1pt]
\hhline{--}
\end{tabular}

\caption{Time complexities of the greedy colouring algorithms we compare with \textsc{Sorted Insertion} in Sec.~\ref{sec:numerics} \textit{after} pre-generating the commutation graph~\cite{Kosowski2004ClassicalCO}.}
\label{tab:complexity}
\end{table}

Note that \textsc{Sorted Insertion} is a heuristic algorithm for maximising $\hat{R}$. It works well in practice, as demonstrated by Table~\ref{tab:cliques} in Sec.~\ref{sec:numerics}, but may fail to output collections that actually maximise $\hat{R}$. In fact, it is unlikely we can go beyond heuristics as the problem of maximising $\hat{R}$ is NP-hard in general. We can show this by combining the reduction from \textsc{Min-Clique-Cover} in Ref.~\cite[Appendix~A]{gokhale2019_commute} with the NP-hardness of $|V|^{1-\epsilon}$-approximating the \textsc{Min-Clique-Cover}~\cite{zuckerman_2006}, where $|V|$ is the number of vertices and $\epsilon=0.5$. We appeal to the hardness of \emph{approximation} because, as we have stressed, maximising $\hat{R}$ is not the same as minimising the number of collections, which corresponds to minimising number of cliques in a cover. For a full proof, see Appendix~\ref{app:rhat_nphard}.

Lastly, we mention that our choice not to put a single Pauli into multiple collections by splitting its coefficient may be sub-optimal as can be seen by considering the operator $O$, on two qubits, defined as
\begin{equation}
    O = X_1 + Z_1 +  2 Z_2
\end{equation}
and a grouping of the form
\begin{equation}
    \mathcal{G}_3 \coloneqq \{\{X_1+\alpha Z_2\}, \{Z_1+(2-\alpha) Z_2\}\}
\end{equation}
for $\alpha \in [0,2]$. It can be verified that $\hat{R}$, defined as before even with coefficient splitting, is maximised at $\alpha = 1$. We leave maximising $\hat{R}$ with consideration of coefficient splitting for future work.

\section{Rotation constructions}\label{sec:rotation_circuit_constructions}
In this section, we present two methods of calculating a rotation circuit which enables measurements of all operators in a mutually commuting collection to be obtained simultaneously. We assume familiarity of the reader with the stabiliser formalism, especially the $2n$-bit binary representation of $n$-qubit Paulis~\cite{qec_binary_orthogonal_geometry, gottesman_phd, nielsen_chuang, aaronson_gottesman_simulation}. We follow the convention that the upper and lower halves of the binary matrix encode $Z$ and $X$ operators respectively. This representation is reviewed in Appendix~\ref{app:binrep}. We also reserve symbols $I_{m}$ and $0_{m}$ for the $m\times m$ identity and all-zero matrices respectively.

Our starting point is a commuting collection, $\mathcal{S}'_\text{start}$, of $m$ Paulis which can be represented as a binary $2n\times m$ matrix $\Sstart'$. By Gaussian elimination on $\Sstart'$, we can form a $2n\times k$ matrix $\Sstart$ representing a set $\mathcal{S}_\text{start}$ of $k$ independent Paulis drawn from  $\mathcal{S}'_\text{start}$ where $k \leq \min(n, m)$. Our goal is to transform $\Sstart$, using certain allowed transformations, into a $2n \times n$ matrix $\Send$ where
\begin{equation}\label{eq:send}
    \Send = 
    \begin{pmatrix}
        \\[1pt]
        I_{n}
        \\[1pt]
        \hline
        \\[1pt]
        0_{n}
        \\[1pt]
    \end{pmatrix}.
\end{equation}

Let $U$ denote the circuit consisting of one-qubit and two-qubit transformations in the order they were applied from $\Sstart \rightarrow \Send$. Then applying $U$ to any state $\ket{\psi}$, measuring in the computational basis, and classically post-processing allows us to obtain measurements of all operators in $\mathcal{S}'_{\text{start}}$ on $\ket{\psi}$ simultaneously. 

The allowed set, $\mathcal{T}$, of transformations on a binary $2n \times m$ matrix $S$ is, where $p$ ranges over all columns, $r$ ranges over all rows, and addition is $\mathrm{mod}\ 2$:
\begin{enumerate}
    \item One-qubit and two-qubit quantum gates, corresponding to row operations, specifically:
    \begin{itemize}\renewcommand{\labelitemi}{\tiny$\blacksquare$}
        \item \textsc{CZ} on qubits $i$ and $j$:\\  $ S_{ip} \leftarrow S_{ip} + S_{j+n, p}$,\\ 
        $S_{jp} \leftarrow S_{jp} + S_{i+n,p}$.
        \item \textsc{CNOT} on control-qubit $i$ and target-qubit $j$: \\
        $ S_{ip} \leftarrow S_{ip} + S_{jp}$,\\ 
        $S_{j+n,p} \leftarrow S_{j+n,p} + S_{i+n,p}$.
        \item \textsc{Hadamard} (\textsc{H}) on qubit $i$:\\
        \hspace{\parindent} $S_{ip} \leftrightarrow S_{i+n,p}$.
        \item \textsc{Phase} (\textsc{P}) on qubit $i$: \\
        \hspace{\parindent} $S_{ip} \leftarrow S_{ip} + S_{i+n,p}$.
    \end{itemize}
    \item Classical post-processing:
    \begin{itemize}\renewcommand{\labelitemi}{\tiny$\blacksquare$}
        \item Products of eventual single-qubit computational-basis measurements:\\
        \hspace{\parindent}right-multiplication by invertible $m \times m$ matrix.
        \item Relabelling of qubits $i$ and $j$:\\
        \hspace{\parindent} $S_{ip} \leftrightarrow S_{jp}$, \\
        \hspace{\parindent} $S_{n+i,p} \leftrightarrow S_{n+j,p}$.
    \end{itemize}
    
    \item Basis extension:
    \begin{itemize}\renewcommand{\labelitemi}{\tiny$\blacksquare$}
    \item Addition of further stabiliser:\\
    \hspace{\parindent}appending of new column $S_{r,m+1}$.
    \end{itemize}
\end{enumerate}

In the near term, operations in $\mathcal{T}$ have justifiably different costs. Two-qubit gates are much more costly than one-qubit gates, which are more costly than classical post-processing. Basis extension has no cost. In the context of one-qubit and two-qubit quantum gates, cost can refer to either fidelity or gate-time~\cite{blais_quantum-information_2007,obrien2010_photon,martinis2014_google,ballance_high-fidelity_2016,allen_optimal_2017,monroe2017_pnas_compare,wendin2017_superconductor_fidelity,reagor2018_rigetti,schafer_fast_2018,webb_resilient_2018,lukin2018_rydberg,simmons2019_electron_spin,huang2019_quantum_dot,blumel_power-optimal_2019,he_two-qubit_2019}. Therefore, we have aimed to minimise the number of two-qubit gates in the $U$ resulting from our constructions. This means, for example, we choose to perform a row swap using classical post-processing rather than the two-qubit \textsc{SWAP} gate. 

In presenting our constructions, we shall refer to the commutativity condition, preserved under $\mathcal{T}$, given by
\begin{equation}\label{eq:commutativitycheck}
    S^\top J_{2n}S  = 0_{m},
\end{equation}
where $S$ is the $2n\times m$ matrix encoding the Paulis and 
\begin{equation}
    J_{2n} = \begin{pmatrix} 
    0_{n} & I_{n}
    \\[1pt]
    \\ I_{n} & 0_{n}
    \end{pmatrix}.
\end{equation}
We ignore any changes in sign of stabilisers under $\mathcal{T}$ as this can be easily accounted for by classical post-processing. Readers interested in this and other details are referred to Appendix~\ref{app:czconstruction}, where we work through our \textsc{CZ}-construction with a specific example.

\subsection{\textsc{CZ}-construction}
\label{sec:czconstr}
Important to our first approach is the special class of stabiliser states known as graph states. Consider any graph $G$ on $n$ vertices. The graph state $\ket{\Phi_G}$ is then defined by $n$ independent stabiliser generators
\begin{equation}\label{eq: graph_stabilisers}
    g_i = X_i \prod_{j \in \text{nbd}(i)} Z_j, \ i = 1, \dots, n,
\end{equation}
where $\text{nbd}(i)$ is the set of neighbours of vertex $i$ in $G$. The binary representation of these stabilisers is
\begin{equation}\label{eq:graph_matrix}
    \Sgraph = 
    \begin{pmatrix}
        \\[1pt]
        A
        \\[1pt]
        \hline
        \\[1pt]
        I_{n}
        \\[1pt]
    \end{pmatrix},
\end{equation}
where $A$, an $n \times n$ symmetric matrix with $0$s on its diagonal, is exactly the adjacency matrix of $G$. 

It is well-known that $\ket{\Phi_G} = V \textsc{H}^{\otimes n}\ket{0^{n}} $ where $V$ is a product of \textsc{CZ} gates and \textsc{H} is the \textsc{Hadamard} gate. More specifically, $V$ applies \textsc{CZ} between qubits $i$ and $j$ if and only if vertex $i$ neighbours $j$ in $G$. \citeauthor*{vandennest2004_graph_states}~\cite{vandennest2004_graph_states} tell us that \textit{any} stabiliser state can be transformed to a graph state by a product of single-qubit Clifford gates and classical post-processing. It is therefore clear that we can transform any $\Sstart$ to $\Send$ via $\Sgraph$ using at most $n(n-1)/2$ two-qubit (\textsc{CZ}) gates, as this is the maximum number of edges on an $n$-vertex graph. The interesting question is whether we can do better by exploiting the potential low rank $k\leq n$ of $\Sstart$.

Our answer is in the affirmative and we now present an explicit and efficient algorithm that constructs $U$ with at most $u_{\text{cz}}(k,n) = kn - k(k+1)/2$ two-qubit gates.

\begin{figure}[ht]
\centering
\resizebox{8.5cm}{!}{%
\begin{tabular}{ccccccc}
   &&  $S_{1}$ &&  $S_{2}$ & & 
    \\[10pt]
    \\
    $\Sstart$ & $\xrightarrow{1}$ & $\begin{pmatrix}
        \\[1pt]
        A
        \\[1pt]
        \hline
        \\[1pt]
        B
        \\[1pt]
    \end{pmatrix}$ & $\xrightarrow{2}$ &
    $\begin{pmatrix}
        \\[1pt]
        C
        \\[1pt]
        D
        \\[1pt]
        \hline
        \\[1pt]
        I_{k}
        \\[1pt]
        F
        \\[1pt]
    \end{pmatrix}$ & 
    \\
    \\
    & & & & & & \\
    \\[10pt]
    && $S_{3}$ &&  $S_{4}=\Sgraph$ &&  
    \\[10pt]
    \\
     & $\xrightarrow{3}$ &     $\begin{pmatrix}
        \\[1pt]
        C & D\trans
        \\[1pt]
        D & 0_{n-k}
        \\[1pt]
        \hline
        \\[1pt]
        I_{k} & 0
        \\[1pt]
        F & I_{n-k}
        \\[1pt]
    \end{pmatrix}$ & $\xrightarrow{4}$ &
    $\begin{pmatrix}
        \\[1pt]
        E & D\trans
        \\[1pt]
        D & 0_{n-k}
        \\[1pt]
        \hline
        \\[1pt]
        I_{k} & 0
        \\[1pt]
        0 & I_{n-k}
        \\[1pt]
    \end{pmatrix}$ & $\xrightarrow{5}$ & 
    $\Send$
\end{tabular}}
\caption{{Reductions used in our \textsc{CZ}-construction.}}
\label{fig:graph_state_reduction}
\end{figure}

In Fig.~\ref{fig:graph_state_reduction}, we illustrate the sequence of reductions that allow us to reach $\Sgraph$, and so $\Send$, from $\Sstart$. We now describe the salient aspects of each step.

\begin{enumerate}
    \item $\Sstart \rightarrow S_{1}.$ Following Ref.~\cite[Lemma 6]{aaronson_gottesman_simulation}, we can apply \textsc{Hadamard} gates so that $B$ has rank $k$. By classical row-swaps (relabelling of qubits), we can ensure that the first $k$ rows of $B$ have full-rank.
    
    \item $S_{1}\rightarrow S_{2}.$ Since the upper $k\times k$ submatrix of $B$ has full-rank, column operations corresponding to classical post-processing can reduce it to $I_{k}$.
    
    \item $S_{2}\rightarrow S_{3}.$ Additional columns, corresponding to further operators, are appended. We can directly verify that the extension to $S_{3}$ is valid by Eq.~\eqref{eq:commutativitycheck}. Clearly $S_{3}$ has full column-rank $n$. The sparsity of our chosen extension shall play a crucial role in the reduced two-qubit gate size of $U$ when $k<n$ in both the \textsc{CZ}- and \textsc{CNOT}-constructions. 
    
    \item $S_{3}\rightarrow S_{4}=\Sgraph.$ Column operations can eliminate $F$, then \textsc{Phase} gates can ensure $E$ has zeros on its diagonal. Importantly, $S_{4}$ represents a graph state $\Sgraph$.
    
    \item $S_{4}\rightarrow \Send$. \textsc{Phase} and \textsc{CZ} gates can implement this final reduction as discussed above. The maximum number of \textsc{CZ} gates required to map $S_{4}$ to $\Send$ equals the maximum number of off-diagonal $1$s in the upper half of $S_{4}$. When $n=k$, this is $n(n-1)/2 = O(n^{2})$. When $k\neq n$, this is $w_{cz}(k,n) = kn - k(k+1)/2$ due to sparsity of the upper half of $S_{4}$ which traces back to the structure of the additional operators appended in step $S_{3}\rightarrow S_{4}$. 
\end{enumerate}

Note that in step $S_{4}\rightarrow \Send$, we can first try to reduce the upper half of $S_{4}$ by single-qubit gates before applying \textsc{CZ}. One way to do this is to reduce the number of edges in the graph whose adjacency matrix is specified by the upper half of $S_{4}$ by the so-called ``local complementation'' operation ~\cite{vandennest2004_graph_states,eisert2004_graph_lu, bouchet2002_lu_graphs}. This corresponds precisely to reducing the number of $\textsc{CZ}$ gates in our $\textsc{CZ}$-construction.

\subsection{\textsc{CNOT}-construction}
\label{sec:cnotconstr}

\begin{figure}[ht]
\centering
\resizebox{8.5cm}{!}{%
\begin{tabular}{ccccccc}
    &&  $S_{5}$ &&   $S_{6}$ & & 
    \\[10pt]
    \\
    $S_{4}$ & $\xrightarrow{5}$ & $\begin{pmatrix}
        \\[1pt]
        0_{k} & D\trans
        \\[1pt]
        D & 0_{n-k}
        \\[1pt]
        \hline
        \\[1pt]
        M & 0
        \\[1pt]
        0 & I_{n-k}
        \\[1pt]
    \end{pmatrix}$ & $\xrightarrow{6}$ &
    $\begin{pmatrix}
        \\[1pt]
        I_{k} & D\trans
        \\[1pt]
        D & I_{n-k}
        \\[1pt]
        \hline
        \\[1pt]
        I_{k} & 0
        \\[1pt]
        0 & I_{n-k}
        \\[1pt]
    \end{pmatrix}$ & 
    \\
    \\
    & & & & & & \\
    [10pt]
     && $S_{7}$ &&  $S_{8}$ &&  
    \\[10pt]
    \\
     & $\xrightarrow{7}$ &     $\begin{pmatrix}
        \\[1pt]
        D_{1} M_{1}
        \\[1pt]
        \hline
        \\[1pt]
        M_{1}
        \\[1pt]
    \end{pmatrix}$ & $\xrightarrow{8}$ &
    $\begin{pmatrix}
        \\[1pt]
        I_{k}-D\trans D & 0
        \\[1pt]
        0 & I_{n-k} \\
        \\[1pt]
        \hline
        \\[1pt]
        I_{k} & 0 \\
        \\[1pt]
        0 & I_{n-k}
        \\[1pt]
    \end{pmatrix}$ & $\xrightarrow{9}$ & 
    $\Send$
\end{tabular}}
\caption{{Reductions used in our \textsc{CNOT}-construction starting at $S_{4}$ of our \textsc{CZ}-construction.}}
\label{fig:cnot_block_reduction}
\end{figure}

We start from $S_{4}$ above, which we reached without using two-qubit gates. Now, instead of using one block of \textsc{CZ} gates, we reduce to $\Send$ as shown in Fig.~\ref{fig:cnot_block_reduction}, using three blocks of \textsc{CNOT} gates.

\begin{enumerate}
\setcounter{enumi}{4}
    \item $S_{4}\rightarrow S_{5}$. Note that $E$ must be symmetric by the commutativity condition given in  Eq.~\eqref{eq:commutativitycheck}. Then, following Ref.~\cite[Lemma 7]{aaronson_gottesman_simulation}, we can eliminate $E$ using single-qubit and $O(k^{2}/ \log k)$ \textsc{CNOT} gates. This is accomplished by noting that any symmetric binary $E$ can be Cholesky decomposed as $E=\Lambda + M\trans M$, with $\Lambda$ diagonal and $M$ invertible.
    \item $S_{5}\rightarrow S_{6}$. Reduce $M$ to $I_{k}$ by column operations, then add $1$s on the top diagonal by phase gates.
    \item $S_{6} \rightarrow S_{7}$. Now, the upper $n\times n$ matrix can be block-Cholesky decomposed as
\begin{equation}\label{eq: block-cholesky}
\begin{pmatrix}
        I_{k} & D\trans
        \\[1pt]
        D & I_{n-k}
\end{pmatrix}
=
M_{1}\trans D_{1}M_{1},
\end{equation} 
where
\begin{align}
    M_{1} &\coloneqq 
    \begin{pmatrix}
    I_{k} & 0 
    \\[1pt]
    D & I_{n-k}
\end{pmatrix},\\
    D_{1} &\coloneqq
    \begin{pmatrix}
        I_{k} - D\trans D & 0 
        \\[1pt]
        0 & I_{n-k}
    \end{pmatrix}.
\end{align}

Next, we apply \textsc{CNOT} gates corresponding to $M_{1}$. The number of \textsc{CNOT} gates required here equals the number of row operations required to reduce $M_{1}$ to $I_{n}$. We find this is at most $u_{\text{cnot}}(k,n)=O(kn/\log k)$ via arguments of~\citeauthor*{patel_markov_hayes_cnot}~\cite{patel_markov_hayes_cnot}. The proof is given in Appendix~\ref{app:fourrussians}.

    \item $S_{7}\rightarrow S_{8}$. Multiply by $M_{1}^{-1}$ on the right.
    
    \item $S_{8}\rightarrow \Send$. $I_{k}-D\trans D$ is a $k\times k$ symmetric matrix and so can be again eliminated via the Cholesky decomposition using $O(k^{2}/\log k)$ \textsc{CNOT} gates.
\end{enumerate}

Note that in the three steps $S_{4}\rightarrow S_{5}$, $S_{6}\rightarrow S_{7}$, and $S_{8}\rightarrow \Send$, we have used blocks of \textsc{CNOT} gates. The method we used to synthesise these blocks is size-optimal~\cite[Lemma~1]{patel_markov_hayes_cnot}, but we could have alternatively used methods in Ref.~\cite{2019optimal_cnot_space_depth_cas_pku}, that built on Ref.~\cite{1998moore_nilsson_cnot_depth}, to achieve optimal space-depth tradeoff, where space refers to extra ancilla qubits.

To end our discussion of constructing rotation circuits, we briefly mention a third, ancilla-based construction with two-qubit gate size at most $kn$. This construction is well-known in the context of syndrome measurement~\cite{kitaev_qec_measure,nielsen_chuang} in quantum error correction but does not seem to have been mentioned in the context of measuring a Pauli decomposition of an operator, as in VQE. To measure $k$ commuting Paulis $\{P_{i}\}_{i=1}^{k}$, this ``ancilla-construction'' uses $k$ ancilla and involves $k$ consecutive blocks of generalised-\textsc{CNOT} gates, each targeted at a different ancilla. The controls in block $b\leq k$ are activated or deactivated by the $+1$ or $-1$ eigenstates of the single-qubit Paulis forming $P_{b}$\footnote{So $Z$ corresponds to standard control on $\{\ket{0}, \ket{1}\}$, $X$ corresponds to control on $\{\ket{+}, \ket{-}\}$, $Y$ corresponds to control on $\{\ket{+i}, \ket{-i}\}$, and $I$ corresponds to not having a generalised-\textsc{CNOT}. These generalised-\textsc{CNOT} gates can be implemented by the standard \textsc{CNOT} conjugated by $H$ in case of $X$, or $SH$ in case of $Y$.}. $k$ single-qubit measurements are performed on the ancilla at the end of each block to exactly give measurements of $P_{i}$. Unfortunately, this construction requires $k$ extra ancilla qubits (or else a single extra ancilla qubit that needs to be measured and reset $k$ times) and has worse worst-case two-qubit gate size than both of our constructions.

\section{Application to VQE}\label{sec:numerics}

In this section, we present numerical results of the collecting method discussed in Sec.~\ref{sec:collections}, alongside the \textsc{CZ}-construction of Sec.~\ref{sec:czconstr} to construct the rotation circuits for given commuting collections. In particular, we have applied our methods to the Hamiltonians of simple molecules so as to
demonstrate their use in the context of VQE. The full results are given in Table~\ref{tab:fulldat} of Appendix~\ref{app:full_numerical_results}, with a subset shown in Table~\ref{tab:reddat}. In all cases, we used OpenFermion~\cite{openfermion} to obtain Hamiltonians in the STO-3G  basis, at approximately the equilibrium geometry of the molecules, with the symmetry conserving Bravyi-Kitaev transformation~\cite{bravyi_kitaev2002, bksc_fermion}. In order to reduce the number of two-qubit gates required, we considered qubits on which all operators in a collection locally commute separately -- a one-qubit rotation per locally commuting qubit is all that is required to do so. 

In Fig.~\ref{fig:plots}(a), we plot the average collection size against the number of qubits, $n$, for the molecular Hamiltonians. We can see that the average collection size increases with increasing $n$, and the increase does not appear to be slowing down. We therefore conclude that our sorting method works well on systems up to at least size $n = 38$, and looks likely to work for larger systems. However, the key advantage of assembling a Hamiltonian into collections of mutually commuting operators is a reduction in the number of measurements required to obtain an energy expectation to a certain level of accuracy, and collection size alone does not directly quantify this reduction. For a given Hamiltonian and quantum state, the reduction is instead given by $R$, as in Eq.~\eqref{eq:rratio}.

\begin{table*}[ht]\centering\footnotesize
\renewcommand{\arraystretch}{1} %
\resizebox{17cm}{!}{
\Large
\begin{tabular}{@{}C{1.5cm}R{4.5em}R{4.5em} R{0.3cm} R{3.5em}R{3.5em}R{3.5em} R{0.3cm} R{3.5em}R{3.5em}R{3.5em}R{3.5em} R{0.3cm} R{3.5em}R{3.5em}R{3.5em}R{3.5em}R{3.5em}
@{}}\toprule

\rule{0pt}{1.5em} \multirow{2}{*}[-0.5em]{Molecule} & \hspace{1.2em} \multirowcell{2}[-0.7em]{\hfill $n$\\ \hfill qubits} & \hspace{1em} \multirowcell{2}[-0.5em]{\hfill $t$\\ \hfill Paulis} &  &\multicolumn{3}{c}{Arrangement} & \hspace{1em} &\multicolumn{4}{c}{Ratios $R$, $\hat{R}$} & \hspace{1em}
 \hspace{1em}&\multicolumn{3}{c}{Rotation Circuit $2q$-size}\\ 

\hhline{~~~~---~----~---}\noalign{\vskip-0.5pt}
\rule{0pt}{1.5em} &  &  & \hspace{1em} & $N$ & $ \overline{m_{i}}$ & $\overline{k_{i}}$ & \hspace{1em} & $R$ min & $R$ mean & $R$ max & $\hat{R}$ & \hspace{1em} & theory max & true max & mean\\ 

\hhline{----------------}
H$_{2}$ & 2 & 4 & \hspace{1em}      & 2 & 2.00 & 1.50 & \hspace{1em} & 1.09 & 1.93 & 4.60 & 1.76 & \hspace{1em} & 0 & 0 & 0\\
H$_{3}^{+}$ & 4 & 59 & \hspace{1em} & 10 & 5.90 & 3.50 & \hspace{1em} & 3.76 & 11.92 & 33.04 & 10.25 & \hspace{1em} & 6 & 3 & 0.80\\
LiH & 10 & 630 & \hspace{1em}       & 41 & 15.37 & 6.85 & \hspace{1em} & 19.60 & 24.91 & 34.74 & 23.97 & \hspace{1em} & 45 & 18 & 5.29\\
OH$^{-}$ & 10 & 630 & \hspace{1em}  & 38 & 16.58 & 7.29 & \hspace{1em} & 6.32 & 8.90 & 12.86 & 8.51 & \hspace{1em} & 45 & 17 & 5.63\\
HF & 10 & 630 & \hspace{1em}        & 39 & 16.15 & 6.97 & \hspace{1em} & 6.07 & 8.57 & 12.27 & 8.21 & \hspace{1em} & 45 & 16 & 5.74\\
H$_{2}$O & 12 & 1085 & \hspace{1em} & 51 & 21.27 & 9.04 & \hspace{1em} & 7.68 & 11.27 & 16.96 & 10.67 & \hspace{1em} & 66 & 26 & 7.37\\
BH$_{3}$ & 14 & 1584 & \hspace{1em} & 66 & 24.00 & 10.36 & \hspace{1em} & 17.21 & 20.93 & 32.13 & 20.05 & \hspace{1em} & 91 & 26 & 9.56\\
NH$_{3}$ & 14 & 3608 & \hspace{1em} & 118 & 30.58 & 11.34 & \hspace{1em} & 12.65 & 15.96 & 26.93 & 15.31 & \hspace{1em} & 91 & 28 & 10.26\\
CH$_{4}$ & 16 & 3887 & \hspace{1em} & 123 & 31.60 & 13.39 & \hspace{1em} & 16.96 & 21.63 & 29.33 & 20.27 & \hspace{1em} & 120 & 45 & 16.75\\
\toprule
\end{tabular}}
\caption{{A reduced set of results of the numerical simulations discussed in the main text, and shown in full in Table~\ref{tab:fulldat}. For each molecule, we show a number of results related to the collecting of Hamiltonian terms, how the arrangement reduces the number of measurements required using the metrics $R$ and $\hat{R}$, and the number of two-qubit gates in the resulting rotation circuits. Note that the mean value of $R$ and $\hat{R}$ are very similar.}}
\label{tab:reddat}
\end{table*}

\begin{table*}[ht]
\centering
\renewcommand{\arraystretch}{1} %
\resizebox{17cm}{!}{%
\Large
\begin{tabular}{@{}C{1.5cm} R{0.3cm} R{3.5em}R{3.5em} R{0.3cm} R{3.5em}R{3.5em} R{0.3cm} R{3.5em}R{3.5em} R{0.3cm} R{3.5em}R{3.5em} R{0.3cm} R{3.5em}R{3.5em}
@{}}\toprule

\rule{0pt}{1.5em} \multirow{2}{*}{Molecule} & \hspace{1em} &\multicolumn{2}{c}{Largest First} & \hspace{1em} &\multicolumn{2}{c}{Connected Sequential d.f.s.} & \hspace{1em} &\multicolumn{2}{c}{Independent Set } & \hspace{1em}
&\multicolumn{2}{c}{DSATUR} & \hspace{1em}
&\multicolumn{2}{c}{\textsc{Sorted Insertion}} \\ 

\hhline{~~--~--~--~--~--}\noalign{\vskip-0.5pt}
\rule{0pt}{1.5em} & \hspace{1em} & $N$ & $\hat{R}$ & \hspace{1em} & $N$ & $\hat{R}$ & \hspace{1em} & $N$ & $\hat{R}$ & \hspace{1em} & $N$ & $\hat{R}$ & \hspace{1em} & $N$ & $\hat{R}$\\
\hhline{----------------}
H$_{2}$ & \hspace{1em} & 2 & \textbf{1.76} & \hspace{1em} & 2 & \textbf{1.76} & \hspace{1em} & 2 & \textbf{1.76} & \hspace{1em} & 2 & \textbf{1.76} & \hspace{1em} & 2 & \textbf{1.76} \\
H$_{3}^{+}$ & \hspace{1em} & 10 & 4.86 & \hspace{1em} & 10 & 10.25 & \hspace{1em} & 10 & \textbf{10.30} & \hspace{1em} & 9 & 4.10 & \hspace{1em}  & 10 & 10.25 \\
LiH & \hspace{1em} & 39 & 23.87 & \hspace{1em} & 45 & 23.33 & \hspace{1em} & 30 & 5.72 & \hspace{1em} & 29 & 10.47 & \hspace{1em} & 41 & \textbf{23.97} \\
OH$^{-}$ & \hspace{1em} & 40 & 8.27 & \hspace{1em} & 41 & 8.41 & \hspace{1em} & 21 & 3.00 & \hspace{1em} & 28 & 3.23 & \hspace{1em} & 37 & \textbf{8.51} \\
HF & \hspace{1em} & 38 & 8.05 & \hspace{1em} & 41 & 8.07 & \hspace{1em} & 21 & 2.80 & \hspace{1em} & 28 & 3.23 & \hspace{1em} & 38 & \textbf{8.21} \\
H$_2$O & \hspace{1em} & 57 & 2.98 & \hspace{1em} & 55 & \textbf{10.66} & \hspace{1em} & 42 & 3.87 & \hspace{1em} & 51 & 3.18 & \hspace{1em} & 51 & \textbf{10.66} \\
BH$_3$ & \hspace{1em} & 66 & 4.80 & \hspace{1em} & 85 & 18.70 & \hspace{1em} & 60 & 7.85 & \hspace{1em} & 72 & 4.11 & \hspace{1em} & 68 & \textbf{20.05} \\
NH$_3$ & \hspace{1em} & 124 & 6.50 & \hspace{1em} & 174 & 13.97 & \hspace{1em} & 126 & 4.03 &  & 137 & 2.92 & \hspace{1em} & 117 & \textbf{15.31} \\
CH$_4$ & \hspace{1em} & 122 & 5.84 & \hspace{1em} & 176 & 18.93 & \hspace{1em} & 114 & 9.88 & \hspace{1em} & 110 & 4.90 & \hspace{1em} & 125 & \textbf{20.27} \\
O$_2$ & \hspace{1em} & 62 & 13.62 & \hspace{1em} & 85 & 19.95 & \hspace{1em} & 42 & 6.79 & \hspace{1em} & 52 & 7.91 & \hspace{1em} & 67 & \textbf{20.23} \\
N$_2$ & \hspace{1em} & 62 & 15.00 & \hspace{1em} & 86 & 21.15 & \hspace{1em} & 39 & 8.37 & \hspace{1em} & 49 & 5.80 & \hspace{1em} & 78 & \textbf{22.10} \\
CO & \hspace{1em} & 124 & 20.70 & \hspace{1em} & 155 & 20.67 & \hspace{1em} & 89 & 6.03 & \hspace{1em} & 106 & 4.55 & \hspace{1em} & 128 & \textbf{21.31} \\
HCl & \hspace{1em} & 117 & 2.16 & \hspace{1em} & 141 & 10.29 & \hspace{1em} & 98 & 3.52 & \hspace{1em} & 104 & 2.04 & \hspace{1em} & 123 & \textbf{10.36} \\
NaH & \hspace{1em} & 121 & 8.78 & \hspace{1em} & 181 & 12.40 & \hspace{1em} & 149 & 3.44 & \hspace{1em} & 145 & 3.65 & \hspace{1em} & 135 & \textbf{12.90} \\
H$_2$S & \hspace{1em} & 122 & 8.81 & \hspace{1em} & 180 & \textbf{12.45} & \hspace{1em} & 147 & 3.80 & \hspace{1em} & 145 & 3.66 & \hspace{1em} & 147 & 11.60 \\
\toprule
\end{tabular}
}
\caption{{Comparison of the arrangements produced by the greedy colouring algorithms ``Largest First'', ``Connected Sequential d.f.s.'' (depth first search), ``Independent Set'' and ``DSATUR'' as implemented by the Python package NetworkX~\cite{hagberg_exploring_2008} with our method \textsc{Sorted Insertion}. For each method, the number of collections produced, $N$, and the metric $\hat{R}$ given by Eq.~\eqref{eq:rhatratio}, are presented. The best or joint best methods are highlighted in bold for each molecule.}}
\label{tab:cliques}
\end{table*}

\begin{figure}[ht]\centering
    \scalebox{0.96}{\input{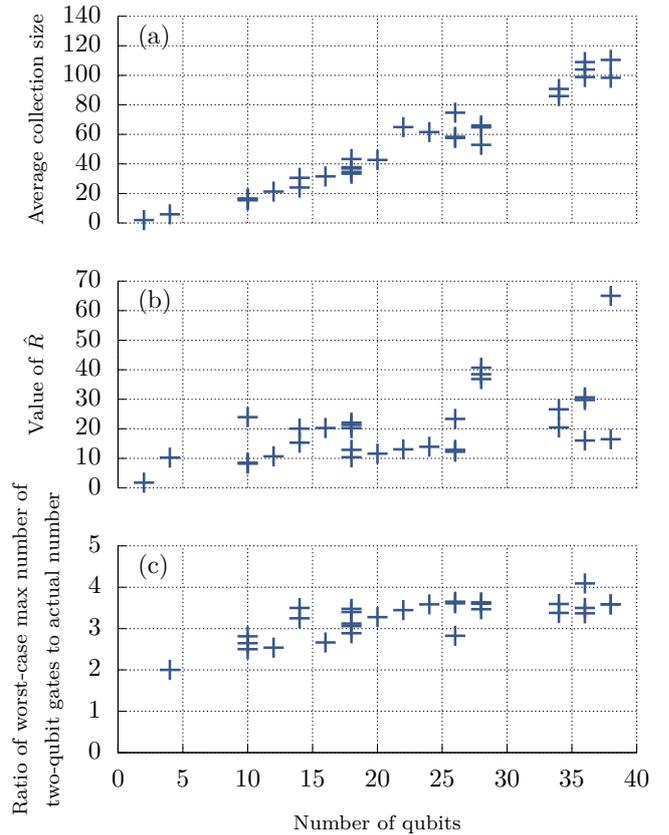}}
    \caption{{The results of numerical simulations discussed in the text. We show (a) the average collection size, (b) the value of $\hat{R}$, given by Eq.~\eqref{eq:rhatratio}, and (c) the ratio of the worst-case maximum number of two-qubit gates in a single rotation circuit to the actual number as a function of the number of qubits for a range of simple molecules. A portion of the data is shown in Table~\ref{tab:reddat}, and the full data is shown in Table~\ref{tab:fulldat}.}}   \label{fig:plots}
\end{figure}

We therefore calculated the value of $R$ for 100 different quantum states, generated using 100 random sets of ansatz parameters with a hardware efficient ansatz of depth 1, for the nine smallest molecular systems. We show the mean, minimum and maximum values for each molecule. In practice, the value of $R$ can at best be obtained approximately by making measurements on the quantum computer and so cannot be used to determine the expected advantage of a particular arrangement a priori. The metric $\hat{R}$, given by Eq.~\eqref{eq:rhatratio}, on the other hand, depends only on the coefficients of the terms in the Hamiltonian. From Table~\ref{tab:reddat}, we can see that $\hat{R}$ closely approximates the average of $R$ over many ansatz parameters, for the ansatz we have considered, but can be calculated analytically without the need for simulations. Further investigation of the relationship between $\hat{R}$ and $R$ will be useful if considering a different ansatz. In Fig.~\ref{fig:plots}(b), we show $\hat{R}$ as a function of the number of qubits for our full selection of molecules. We can see that it is highly molecule dependent, with systems of similar size having very different values.

The reduction in the number of measurements required comes at the cost of applying additional quantum gates before the qubits are measured, the most costly of which are two-qubit gates. For the \textsc{CZ}-construction, we demonstrated in Sec.~\ref{sec:czconstr} that the maximum number of additional two-qubit gates required for a collection with $k$ independent terms is $nk - k(k+1)/2$. We would like to know, in practice, how many additional two-qubit gates are required at a maximum, as this is the quantum resource that is most limiting. Assuming for a given Hamiltonian that at least one collection has rank $n$, obtaining a measurement of all terms in a Hamiltonian on $n$ qubits may therefore require applying an additional $n(n-1)/2$ gates in a single circuit. However, for the molecules we have considered, we find that the largest number of two-qubit gates required is in fact far lower than this, typically by a factor of about 3.5, as can be seen in Fig.~\ref{fig:plots}(c).

Given the close relationship between the average value of $R$ and the value of $\hat{R}$, we propose using $\hat{R}$ as a metric for the quality of a collecting method and compare different methods of collecting the operators with this metric in mind. The results are shown in Table~\ref{tab:cliques}, along with the number of collections of operators, $N$, that each method produces. Out of the methods, ``Independent Set'' was best at approximating the minimum clique cover -- it found the cover with the fewest cliques in all but one case. However, the minimum clique cover does not necessarily result in the fewest measurements, with ``Independent Set'' only performing best once with respect to $\hat{R}$. Overall, our \textsc{Sorted Insertion} was best at maximising $\hat{R}$, performing best or joint best in all but two cases.

\section{Conclusion}
We have addressed two problems related to the efficient measurement of Pauli operators on a quantum computer in the presence of finite sampling error. The first is how to assemble a set of Paulis into collections in which they mutually commute, and the second is how to construct rotation circuits that enable mutually commuting Paulis to be measured simultaneously.

For the first problem, we contribute two natural metrics, $R$ and $\hat{R}$, that justifiably measure the effectiveness of an arrangement, followed by a collecting strategy motivated by $\hat{R}$ that we call \textsc{Sorted Insertion}. For the second problem, we contribute two rotation circuit constructions, \textsc{CZ} and \textsc{CNOT}. The \textsc{CZ}-construction uses a maximum of $u_{\text{cz}}(k,n) = kn - k(k+1)/2$ two-qubit gates while the \textsc{CNOT}-construction uses a maximum of $u_{\text{cnot}}(k,n) = O(kn/\log k)$. 

We have applied our theoretical work to the task of estimating energies of molecules in the context of VQE. Comparison to other collecting methods shows that while \textsc{Sorted Insertion} does not normally result in the smallest number of collections, it nearly always results in the best value of $\hat{R}$. We also find that, for the CZ-construction, the largest number of two-qubit gates required is typically less than the theoretical worst-case by a factor of about 3.5.

\section{Acknowledgement}
We thank Oscar Higgott for informing us of the ancilla-based simultaneous measurement method and Ref.~\cite[Fig.~14]{kitaev_qec_measure}. We also thank the editor Ronald de Wolf and two anonymous reviewers for their comments and suggestions.

\bibliographystyle{unsrtnat}
\bibliography{references}

\begin{thebibliography}{66}
\providecommand{\natexlab}[1]{#1}
\providecommand{\url}[1]{\texttt{#1}}
\expandafter\ifx\csname urlstyle\endcsname\relax
  \providecommand{\doi}[1]{doi: #1}\else
  \providecommand{\doi}{doi: \begingroup \urlstyle{rm}\Url}\fi

\bibitem[Peruzzo et~al.(2014)Peruzzo, McClean, Shadbolt, Yung, Zhou, Love,
  Aspuru-Guzik, and O'Brien]{peruzzo2014variational}
Alberto Peruzzo, Jarrod McClean, Peter Shadbolt, Man-Hong Yung, Xiao-Qi Zhou,
  Peter~J. Love, Al{\'a}n Aspuru-Guzik, and Jeremy~L. O'Brien.
\newblock A variational eigenvalue solver on a photonic quantum processor.
\newblock \emph{Nature Communications}, 5\penalty0 (1):\penalty0 4213, 2014.
\newblock \doi{10.1038/ncomms5213}.

\bibitem[Preskill(2018)]{preskill2018quantum}
John Preskill.
\newblock Quantum computing in the {NISQ} era and beyond.
\newblock \emph{{Quantum}}, 2:\penalty0 79, August 2018.
\newblock ISSN 2521-327X.
\newblock \doi{10.22331/q-2018-08-06-79}.

\bibitem[McClean et~al.(2016)McClean, Romero, Babbush, and
  Aspuru-Guzik]{mcclean2016theory}
Jarrod~R McClean, Jonathan Romero, Ryan Babbush, and Al{\'{a}}n Aspuru-Guzik.
\newblock The theory of variational hybrid quantum-classical algorithms.
\newblock \emph{New Journal of Physics}, 18\penalty0 (2):\penalty0 023023,
  February 2016.
\newblock \doi{10.1088/1367-2630/18/2/023023}.

\bibitem[O'Malley et~al.(2016)O'Malley, Babbush, Kivlichan, Romero, McClean,
  Barends, Kelly, Roushan, Tranter, Ding, Campbell, Chen, Chen, Chiaro,
  Dunsworth, Fowler, Jeffrey, Lucero, Megrant, Mutus, Neeley, Neill, Quintana,
  Sank, Vainsencher, Wenner, White, Coveney, Love, Neven, Aspuru-Guzik, and
  Martinis]{omalley2016scalable}
P.~J.~J. O'Malley, R.~Babbush, I.~D. Kivlichan, J.~Romero, J.~R. McClean,
  R.~Barends, J.~Kelly, P.~Roushan, A.~Tranter, N.~Ding, B.~Campbell, Y.~Chen,
  Z.~Chen, B.~Chiaro, A.~Dunsworth, A.~G. Fowler, E.~Jeffrey, E.~Lucero,
  A.~Megrant, J.~Y. Mutus, M.~Neeley, C.~Neill, C.~Quintana, D.~Sank,
  A.~Vainsencher, J.~Wenner, T.~C. White, P.~V. Coveney, P.~J. Love, H.~Neven,
  A.~Aspuru-Guzik, and J.~M. Martinis.
\newblock Scalable quantum simulation of molecular energies.
\newblock \emph{Phys. Rev. X}, 6:\penalty0 031007, July 2016.
\newblock \doi{10.1103/PhysRevX.6.031007}.

\bibitem[Kandala et~al.(2017)Kandala, Mezzacapo, Temme, Takita, Brink, Chow,
  and Gambetta]{kandala2017_hardware}
Abhinav Kandala, Antonio Mezzacapo, Kristan Temme, Maika Takita, Markus Brink,
  Jerry~M. Chow, and Jay~M. Gambetta.
\newblock Hardware-efficient variational quantum eigensolver for small
  molecules and quantum magnets.
\newblock \emph{Nature}, 549\penalty0 (7671):\penalty0 242--246, 2017.
\newblock \doi{10.1038/nature23879}.

\bibitem[McArdle et~al.(2020)McArdle, Endo, Aspuru-Guzik, Benjamin, and
  Yuan]{mcardle_quantum_2018}
Sam McArdle, Suguru Endo, Al\'an Aspuru-Guzik, Simon~C. Benjamin, and Xiao
  Yuan.
\newblock Quantum computational chemistry.
\newblock \emph{Rev. Mod. Phys.}, 92:\penalty0 015003, Mar 2020.
\newblock \doi{10.1103/RevModPhys.92.015003}.

\bibitem[Ryabinkin et~al.(2019)Ryabinkin, Genin, and
  Izmaylov]{ryabinkin_constrained_2018}
Ilya~G. Ryabinkin, Scott~N. Genin, and Artur~F. Izmaylov.
\newblock Constrained variational quantum eigensolver: Quantum computer search
  engine in the {F}ock space.
\newblock \emph{Journal of Chemical Theory and Computation}, 15\penalty0
  (1):\penalty0 249--255, 01 2019.
\newblock \doi{10.1021/acs.jctc.8b00943}.

\bibitem[Romero et~al.(2018)Romero, Babbush, McClean, Hempel, Love, and
  Aspuru-Guzik]{romero2018strategies}
Jonathan Romero, Ryan Babbush, Jarrod~R McClean, Cornelius Hempel, Peter~J
  Love, and Al{\'{a}}n Aspuru-Guzik.
\newblock Strategies for quantum computing molecular energies using the unitary
  coupled cluster ansatz.
\newblock \emph{Quantum Science and Technology}, 4\penalty0 (1):\penalty0
  014008, Oct 2018.
\newblock \doi{10.1088/2058-9565/aad3e4}.

\bibitem[Wang et~al.(2019)Wang, Higgott, and Brierley]{wang2018generalised}
Daochen Wang, Oscar Higgott, and Stephen Brierley.
\newblock Accelerated variational quantum eigensolver.
\newblock \emph{Phys. Rev. Lett.}, 122:\penalty0 140504, Apr 2019.
\newblock \doi{10.1103/PhysRevLett.122.140504}.

\bibitem[McClean et~al.(2017)McClean, Kimchi-Schwartz, Carter, and
  de~Jong]{mcclean2017hybrid}
Jarrod~R. McClean, Mollie~E. Kimchi-Schwartz, Jonathan Carter, and Wibe~A.
  de~Jong.
\newblock Hybrid quantum-classical hierarchy for mitigation of decoherence and
  determination of excited states.
\newblock \emph{Phys. Rev. A}, 95:\penalty0 042308, Apr 2017.
\newblock \doi{10.1103/PhysRevA.95.042308}.

\bibitem[Santagati et~al.(2018)Santagati, Wang, Gentile, Paesani, Wiebe,
  McClean, Morley-Short, Shadbolt, Bonneau, Silverstone, Tew, Zhou,
  O{\textquoteright}Brien, and Thompson]{santagati2018witnessing}
Raffaele Santagati, Jianwei Wang, Antonio~A. Gentile, Stefano Paesani, Nathan
  Wiebe, Jarrod~R. McClean, Sam Morley-Short, Peter~J. Shadbolt, Damien
  Bonneau, Joshua~W. Silverstone, David~P. Tew, Xiaoqi Zhou, Jeremy~L.
  O{\textquoteright}Brien, and Mark~G. Thompson.
\newblock Witnessing eigenstates for quantum simulation of {H}amiltonian
  spectra.
\newblock \emph{Science Advances}, 4\penalty0 (1), 2018.
\newblock \doi{10.1126/sciadv.aap9646}.

\bibitem[Colless et~al.(2018)Colless, Ramasesh, Dahlen, Blok, Kimchi-Schwartz,
  McClean, Carter, de~Jong, and Siddiqi]{colless2018computation}
J.~I. Colless, V.~V. Ramasesh, D.~Dahlen, M.~S. Blok, M.~E. Kimchi-Schwartz,
  J.~R. McClean, J.~Carter, W.~A. de~Jong, and I.~Siddiqi.
\newblock Computation of molecular spectra on a quantum processor with an
  error-resilient algorithm.
\newblock \emph{Phys. Rev. X}, 8:\penalty0 011021, Feb 2018.
\newblock \doi{10.1103/PhysRevX.8.011021}.

\bibitem[{Heya} et~al.(2019){Heya}, {Nakanishi}, {Mitarai}, and
  {Fujii}]{heya_subspace_2019}
Kentaro {Heya}, Ken~M {Nakanishi}, Kosuke {Mitarai}, and Keisuke {Fujii}.
\newblock Subspace variational quantum simulator.
\newblock \emph{arXiv e-prints}, Apr 2019.

\bibitem[Jones et~al.(2019)Jones, Endo, McArdle, Yuan, and
  Benjamin]{jones_variational_2019}
Tyson Jones, Suguru Endo, Sam McArdle, Xiao Yuan, and Simon~C. Benjamin.
\newblock Variational quantum algorithms for discovering {H}amiltonian spectra.
\newblock \emph{Phys. Rev. A}, 99:\penalty0 062304, Jun 2019.
\newblock \doi{10.1103/PhysRevA.99.062304}.

\bibitem[Higgott et~al.(2019)Higgott, Wang, and
  Brierley]{Higgott2019variationalquantum}
Oscar Higgott, Daochen Wang, and Stephen Brierley.
\newblock Variational quantum computation of excited states.
\newblock \emph{{Quantum}}, 3:\penalty0 156, July 2019.
\newblock ISSN 2521-327X.
\newblock \doi{10.22331/q-2019-07-01-156}.

\bibitem[Verteletskyi et~al.(2020)Verteletskyi, Yen, and
  Izmaylov]{izmaylov2019_minclique}
Vladyslav Verteletskyi, Tzu-Ching Yen, and Artur~F. Izmaylov.
\newblock Measurement optimization in the variational quantum eigensolver using
  a minimum clique cover.
\newblock \emph{The Journal of Chemical Physics}, 152\penalty0 (12):\penalty0
  124114, 2020.
\newblock \doi{10.1063/1.5141458}.

\bibitem[{Jena} et~al.(2019){Jena}, {Genin}, and
  {Mosca}]{mosca2019_pauli_partioning}
Andrew {Jena}, Scott {Genin}, and Michele {Mosca}.
\newblock {P}auli partitioning with respect to gate sets.
\newblock \emph{arXiv e-prints}, July 2019.

\bibitem[Yen et~al.(2020)Yen, Verteletskyi, and
  Izmaylov]{izmaylov2019_general_commute}
Tzu-Ching Yen, Vladyslav Verteletskyi, and Artur~F. Izmaylov.
\newblock Measuring all compatible operators in one series of single-qubit
  measurements using unitary transformations.
\newblock \emph{Journal of Chemical Theory and Computation}, 16\penalty0
  (4):\penalty0 2400--2409, 04 2020.
\newblock \doi{10.1021/acs.jctc.0c00008}.

\bibitem[{Huggins} et~al.(2019){Huggins}, {McClean}, {Rubin}, {Jiang}, {Wiebe},
  {Whaley}, and {Babbush}]{babbush2019_chem_commute}
William~J. {Huggins}, Jarrod {McClean}, Nicholas {Rubin}, Zhang {Jiang}, Nathan
  {Wiebe}, K.~Birgitta {Whaley}, and Ryan {Babbush}.
\newblock Efficient and noise resilient measurements for quantum chemistry on
  near-term quantum computers.
\newblock \emph{arXiv e-prints}, July 2019.

\bibitem[{Gokhale} et~al.(2019){Gokhale}, {Angiuli}, {Ding}, {Gui}, {Tomesh},
  {Suchara}, {Martonosi}, and {Chong}]{gokhale2019_commute}
Pranav {Gokhale}, Olivia {Angiuli}, Yongshan {Ding}, Kaiwen {Gui}, Teague
  {Tomesh}, Martin {Suchara}, Margaret {Martonosi}, and Frederic~T. {Chong}.
\newblock Minimizing state preparations in variational quantum eigensolver by
  partitioning into commuting families.
\newblock \emph{arXiv e-prints}, July 2019.

\bibitem[Zhao et~al.(2020)Zhao, Tranter, Kirby, Ung, Miyake, and
  Love]{zhao_2019}
Andrew Zhao, Andrew Tranter, William~M. Kirby, Shu~Fay Ung, Akimasa Miyake, and
  Peter~J. Love.
\newblock Measurement reduction in variational quantum algorithms.
\newblock \emph{Phys. Rev. A}, 101:\penalty0 062322, Jun 2020.
\newblock \doi{10.1103/PhysRevA.101.062322}.

\bibitem[{Gokhale} et~al.(2020){Gokhale}, {Angiuli}, {Ding}, {Gui}, {Tomesh},
  {Suchara}, {Martonosi}, and {Chong}]{gokhale_n3}
P.~{Gokhale}, O.~{Angiuli}, Y.~{Ding}, K.~{Gui}, T.~{Tomesh}, M.~{Suchara},
  M.~{Martonosi}, and F.~T. {Chong}.
\newblock {$O(N^3)$} measurement cost for variational quantum eigensolver on
  molecular {H}amiltonians.
\newblock \emph{IEEE Transactions on Quantum Engineering}, 1:\penalty0 1--24,
  2020.
\newblock \doi{10.1109/TQE.2020.3035814}.

\bibitem[Izmaylov et~al.(2020)Izmaylov, Yen, Lang, and
  Verteletskyi]{2019izmaylov_anticommute}
Artur~F. Izmaylov, Tzu-Ching Yen, Robert~A. Lang, and Vladyslav Verteletskyi.
\newblock Unitary partitioning approach to the measurement problem in the
  variational quantum eigensolver method.
\newblock \emph{Journal of Chemical Theory and Computation}, 16\penalty0
  (1):\penalty0 190--195, 01 2020.
\newblock \doi{10.1021/acs.jctc.9b00791}.

\bibitem[Izmaylov et~al.(2019)Izmaylov, Yen, and
  Ryabinkin]{izmaylov2018_revising}
Artur~F. Izmaylov, Tzu-Ching Yen, and Ilya~G. Ryabinkin.
\newblock Revising the measurement process in the variational quantum
  eigensolver: is it possible to reduce the number of separately measured
  operators?
\newblock \emph{Chem. Sci.}, 10:\penalty0 3746--3755, 2019.
\newblock \doi{10.1039/C8SC05592K}.

\bibitem[Wecker et~al.(2015)Wecker, Hastings, and Troyer]{wecker2015progress}
Dave Wecker, Matthew~B. Hastings, and Matthias Troyer.
\newblock Progress towards practical quantum variational algorithms.
\newblock \emph{Phys. Rev. A}, 92:\penalty0 042303, Oct 2015.
\newblock \doi{10.1103/PhysRevA.92.042303}.

\bibitem[Rubin et~al.(2018)Rubin, Babbush, and McClean]{Rubin_2018}
Nicholas~C Rubin, Ryan Babbush, and Jarrod McClean.
\newblock Application of fermionic marginal constraints to hybrid quantum
  algorithms.
\newblock \emph{New Journal of Physics}, 20\penalty0 (5):\penalty0 053020, May
  2018.
\newblock \doi{10.1088/1367-2630/aab919}.

\bibitem[Blais et~al.(2007)Blais, Gambetta, Wallraff, Schuster, Girvin,
  Devoret, and Schoelkopf]{blais_quantum-information_2007}
Alexandre Blais, Jay Gambetta, A.~Wallraff, D.~I. Schuster, S.~M. Girvin, M.~H.
  Devoret, and R.~J. Schoelkopf.
\newblock Quantum-information processing with circuit quantum electrodynamics.
\newblock \emph{Phys. Rev. A}, 75:\penalty0 032329, Mar 2007.
\newblock \doi{10.1103/PhysRevA.75.032329}.

\bibitem[Laing et~al.(2010)Laing, Peruzzo, Politi, Verde, Halder, Ralph,
  Thompson, and O'Brien]{obrien2010_photon}
Anthony Laing, Alberto Peruzzo, Alberto Politi, Maria~Rodas Verde, Matthaeus
  Halder, Timothy~C. Ralph, Mark~G. Thompson, and Jeremy~L. O'Brien.
\newblock High-fidelity operation of quantum photonic circuits.
\newblock \emph{Applied Physics Letters}, 97\penalty0 (21):\penalty0 211109,
  2010.
\newblock \doi{10.1063/1.3497087}.

\bibitem[Barends et~al.(2014)Barends, Kelly, Megrant, Veitia, Sank, Jeffrey,
  White, Mutus, Fowler, Campbell, Chen, Chen, Chiaro, Dunsworth, Neill,
  O'Malley, Roushan, Vainsencher, Wenner, Korotkov, Cleland, and
  Martinis]{martinis2014_google}
R.~Barends, J.~Kelly, A.~Megrant, A.~Veitia, D.~Sank, E.~Jeffrey, T.~C. White,
  J.~Mutus, A.~G. Fowler, B.~Campbell, Y.~Chen, Z.~Chen, B.~Chiaro,
  A.~Dunsworth, C.~Neill, P.~O'Malley, P.~Roushan, A.~Vainsencher, J.~Wenner,
  A.~N. Korotkov, A.~N. Cleland, and John~M. Martinis.
\newblock Superconducting quantum circuits at the surface code threshold for
  fault tolerance.
\newblock \emph{Nature}, 508\penalty0 (7497):\penalty0 500--503, 2014.
\newblock \doi{10.1038/nature13171}.

\bibitem[Ballance et~al.(2016)Ballance, Harty, Linke, Sepiol, and
  Lucas]{ballance_high-fidelity_2016}
C.~J. Ballance, T.~P. Harty, N.~M. Linke, M.~A. Sepiol, and D.~M. Lucas.
\newblock High-fidelity quantum logic gates using trapped-ion hyperfine qubits.
\newblock \emph{Phys. Rev. Lett.}, 117:\penalty0 060504, Aug 2016.
\newblock \doi{10.1103/PhysRevLett.117.060504}.

\bibitem[Allen et~al.(2017)Allen, Kosut, Joo, Leek, and
  Ginossar]{allen_optimal_2017}
Joseph~L. Allen, Robert Kosut, Jaewoo Joo, Peter Leek, and Eran Ginossar.
\newblock Optimal control of two qubits via a single cavity drive in circuit
  quantum electrodynamics.
\newblock \emph{Phys. Rev. A}, 95:\penalty0 042325, Apr 2017.
\newblock \doi{10.1103/PhysRevA.95.042325}.

\bibitem[Linke et~al.(2017)Linke, Maslov, Roetteler, Debnath, Figgatt,
  Landsman, Wright, and Monroe]{monroe2017_pnas_compare}
Norbert~M. Linke, Dmitri Maslov, Martin Roetteler, Shantanu Debnath, Caroline
  Figgatt, Kevin~A. Landsman, Kenneth Wright, and Christopher Monroe.
\newblock Experimental comparison of two quantum computing architectures.
\newblock \emph{Proceedings of the National Academy of Sciences}, 114\penalty0
  (13):\penalty0 3305--3310, 2017.
\newblock ISSN 0027-8424.
\newblock \doi{10.1073/pnas.1618020114}.

\bibitem[Wendin(2017)]{wendin2017_superconductor_fidelity}
G~Wendin.
\newblock Quantum information processing with superconducting circuits: a
  review.
\newblock \emph{Reports on Progress in Physics}, 80\penalty0 (10):\penalty0
  106001, sep 2017.
\newblock \doi{10.1088/1361-6633/aa7e1a}.

\bibitem[Reagor et~al.(2018)Reagor, Osborn, Tezak, Staley, Prawiroatmodjo,
  Scheer, Alidoust, Sete, Didier, da~Silva, Acala, Angeles, Bestwick, Block,
  Bloom, Bradley, Bui, Caldwell, Capelluto, Chilcott, Cordova, Crossman,
  Curtis, Deshpande, El~Bouayadi, Girshovich, Hong, Hudson, Karalekas, Kuang,
  Lenihan, Manenti, Manning, Marshall, Mohan, O{\textquoteright}Brien,
  Otterbach, Papageorge, Paquette, Pelstring, Polloreno, Rawat, Ryan, Renzas,
  Rubin, Russel, Rust, Scarabelli, Selvanayagam, Sinclair, Smith, Suska, To,
  Vahidpour, Vodrahalli, Whyland, Yadav, Zeng, and Rigetti]{reagor2018_rigetti}
Matthew Reagor, Christopher~B. Osborn, Nikolas Tezak, Alexa Staley, Guenevere
  Prawiroatmodjo, Michael Scheer, Nasser Alidoust, Eyob~A. Sete, Nicolas
  Didier, Marcus~P. da~Silva, Ezer Acala, Joel Angeles, Andrew Bestwick,
  Maxwell Block, Benjamin Bloom, Adam Bradley, Catvu Bui, Shane Caldwell,
  Lauren Capelluto, Rick Chilcott, Jeff Cordova, Genya Crossman, Michael
  Curtis, Saniya Deshpande, Tristan El~Bouayadi, Daniel Girshovich, Sabrina
  Hong, Alex Hudson, Peter Karalekas, Kat Kuang, Michael Lenihan, Riccardo
  Manenti, Thomas Manning, Jayss Marshall, Yuvraj Mohan, William
  O{\textquoteright}Brien, Johannes Otterbach, Alexander Papageorge,
  Jean-Philip Paquette, Michael Pelstring, Anthony Polloreno, Vijay Rawat,
  Colm~A. Ryan, Russ Renzas, Nick Rubin, Damon Russel, Michael Rust, Diego
  Scarabelli, Michael Selvanayagam, Rodney Sinclair, Robert Smith, Mark Suska,
  Ting-Wai To, Mehrnoosh Vahidpour, Nagesh Vodrahalli, Tyler Whyland, Kamal
  Yadav, William Zeng, and Chad~T. Rigetti.
\newblock Demonstration of universal parametric entangling gates on a
  multi-qubit lattice.
\newblock \emph{Science Advances}, 4\penalty0 (2), 2018.
\newblock \doi{10.1126/sciadv.aao3603}.

\bibitem[Sch{\"a}fer et~al.(2018)Sch{\"a}fer, Ballance, Thirumalai, Stephenson,
  Ballance, Steane, and Lucas]{schafer_fast_2018}
V.~M. Sch{\"a}fer, C.~J. Ballance, K.~Thirumalai, L.~J. Stephenson, T.~G.
  Ballance, A.~M. Steane, and D.~M. Lucas.
\newblock Fast quantum logic gates with trapped-ion qubits.
\newblock \emph{Nature}, 555\penalty0 (7694):\penalty0 75--78, 2018.
\newblock \doi{10.1038/nature25737}.

\bibitem[Webb et~al.(2018)Webb, Webster, Collingbourne, Bretaud, Lawrence,
  Weidt, Mintert, and Hensinger]{webb_resilient_2018}
A.~E. Webb, S.~C. Webster, S.~Collingbourne, D.~Bretaud, A.~M. Lawrence,
  S.~Weidt, F.~Mintert, and W.~K. Hensinger.
\newblock Resilient entangling gates for trapped ions.
\newblock \emph{Phys. Rev. Lett.}, 121:\penalty0 180501, Nov 2018.
\newblock \doi{10.1103/PhysRevLett.121.180501}.

\bibitem[Levine et~al.(2018)Levine, Keesling, Omran, Bernien, Schwartz, Zibrov,
  Endres, Greiner, Vuleti\ifmmode~\acute{c}\else \'{c}\fi{}, and
  Lukin]{lukin2018_rydberg}
Harry Levine, Alexander Keesling, Ahmed Omran, Hannes Bernien, Sylvain
  Schwartz, Alexander~S. Zibrov, Manuel Endres, Markus Greiner, Vladan
  Vuleti\ifmmode~\acute{c}\else \'{c}\fi{}, and Mikhail~D. Lukin.
\newblock High-fidelity control and entanglement of {R}ydberg-atom qubits.
\newblock \emph{Phys. Rev. Lett.}, 121:\penalty0 123603, Sep 2018.
\newblock \doi{10.1103/PhysRevLett.121.123603}.

\bibitem[He et~al.(2019{\natexlab{a}})He, Gorman, Keith, Kranz, Keizer, and
  Simmons]{simmons2019_electron_spin}
Y.~He, S.~K. Gorman, D.~Keith, L.~Kranz, J.~G. Keizer, and M.~Y. Simmons.
\newblock A two-qubit gate between phosphorus donor electrons in silicon.
\newblock \emph{Nature}, 571\penalty0 (7765):\penalty0 371--375,
  2019{\natexlab{a}}.
\newblock \doi{10.1038/s41586-019-1381-2}.

\bibitem[Huang et~al.(2019)Huang, Yang, Chan, Tanttu, Hensen, Leon, Fogarty,
  Hwang, Hudson, Itoh, Morello, Laucht, and Dzurak]{huang2019_quantum_dot}
W.~Huang, C.~H. Yang, K.~W. Chan, T.~Tanttu, B.~Hensen, R.~C.~C. Leon, M.~A.
  Fogarty, J.~C.~C. Hwang, F.~E. Hudson, K.~M. Itoh, A.~Morello, A.~Laucht, and
  A.~S. Dzurak.
\newblock Fidelity benchmarks for two-qubit gates in silicon.
\newblock \emph{Nature}, 569\penalty0 (7757):\penalty0 532--536, 2019.
\newblock \doi{10.1038/s41586-019-1197-0}.

\bibitem[{Blumel} et~al.(2019){Blumel}, {Grzesiak}, and
  {Nam}]{blumel_power-optimal_2019}
Reinhold {Blumel}, Nikodem {Grzesiak}, and Yunseong {Nam}.
\newblock Power-optimal, stabilized entangling gate between trapped-ion qubits.
\newblock \emph{arXiv e-prints}, May 2019.

\bibitem[He et~al.(2019{\natexlab{b}})He, Gorman, Keith, Kranz, Keizer, and
  Simmons]{he_two-qubit_2019}
Y.~He, S.~K. Gorman, D.~Keith, L.~Kranz, J.~G. Keizer, and M.~Y. Simmons.
\newblock A two-qubit gate between phosphorus donor electrons in silicon.
\newblock \emph{Nature}, 571\penalty0 (7765):\penalty0 371--375,
  2019{\natexlab{b}}.
\newblock \doi{10.1038/s41586-019-1381-2}.

\bibitem[Van~den Nest et~al.(2004)Van~den Nest, Dehaene, and
  De~Moor]{vandennest2004_graph_states}
Maarten Van~den Nest, Jeroen Dehaene, and Bart De~Moor.
\newblock Graphical description of the action of local {C}lifford
  transformations on graph states.
\newblock \emph{Phys. Rev. A}, 69:\penalty0 022316, Feb 2004.
\newblock \doi{10.1103/PhysRevA.69.022316}.

\bibitem[Aaronson and Gottesman(2004)]{aaronson_gottesman_simulation}
Scott Aaronson and Daniel Gottesman.
\newblock Improved simulation of stabilizer circuits.
\newblock \emph{Phys. Rev. A}, 70:\penalty0 052328, Nov 2004.
\newblock \doi{10.1103/PhysRevA.70.052328}.

\bibitem[Patel et~al.(2008)Patel, Markov, and Hayes]{patel_markov_hayes_cnot}
Ketan~N. Patel, Igor~L. Markov, and John~P. Hayes.
\newblock Optimal synthesis of linear reversible circuits.
\newblock \emph{Quantum Info. Comput.}, 8\penalty0 (3):\penalty0 282--294,
  March 2008.
\newblock ISSN 1533-7146.
\newblock \doi{10.26421/QIC8.3-4}.

\bibitem[Huang et~al.(2020)Huang, Kueng, and Preskill]{huang_2020}
Hsin-Yuan Huang, Richard Kueng, and John Preskill.
\newblock Predicting many properties of a quantum system from very few
  measurements.
\newblock \emph{Nature Physics}, 16\penalty0 (10):\penalty0 1050–1057, Jun
  2020.
\newblock ISSN 1745-2481.
\newblock \doi{10.1038/s41567-020-0932-7}.

\bibitem[{Hadfield} et~al.(2020){Hadfield}, {Bravyi}, {Raymond}, and
  {Mezzacapo}]{hadfield_2020}
Charles {Hadfield}, Sergey {Bravyi}, Rudy {Raymond}, and Antonio {Mezzacapo}.
\newblock Measurements of quantum {H}amiltonians with locally-biased classical
  shadows.
\newblock \emph{arXiv e-prints}, June 2020.

\bibitem[Kübler et~al.(2020)Kübler, Arrasmith, Cincio, and
  Coles]{kubler_2020}
Jonas~M. Kübler, Andrew Arrasmith, Lukasz Cincio, and Patrick~J. Coles.
\newblock An adaptive optimizer for measurement-frugal variational algorithms.
\newblock \emph{Quantum}, 4:\penalty0 263, May 2020.
\newblock ISSN 2521-327X.
\newblock \doi{10.22331/q-2020-05-11-263}.

\bibitem[{Arrasmith} et~al.(2020){Arrasmith}, {Cincio}, {Somma}, and
  {Coles}]{arrasmith_2020}
Andrew {Arrasmith}, Lukasz {Cincio}, Rolando~D. {Somma}, and Patrick~J.
  {Coles}.
\newblock Operator sampling for shot-frugal optimization in variational
  algorithms.
\newblock \emph{arXiv e-prints}, April 2020.

\bibitem[Conway(1994)]{conway1994course}
J.B. Conway.
\newblock \emph{A Course in Functional Analysis}.
\newblock Graduate Texts in Mathematics. Springer New York, 1994.
\newblock ISBN 9780387972459.

\bibitem[Watrous(2018)]{watrous_2018}
John Watrous.
\newblock \emph{The Theory of Quantum Information}.
\newblock Cambridge University Press, 2018.

\bibitem[Kosowski and Manuszewski(2004)]{Kosowski2004ClassicalCO}
Adrian Kosowski and Krzysztof Manuszewski.
\newblock Classical coloring of graphs.
\newblock In Marek Kubale, editor, \emph{Graph Colorings}, chapter~1. American
  Mathematical Society, 2004.
\newblock ISBN 978-0-8218-7942-9.
\newblock \doi{10.1090/conm/352}.

\bibitem[Zuckerman(2006)]{zuckerman_2006}
David Zuckerman.
\newblock Linear degree extractors and the inapproximability of max clique and
  chromatic number.
\newblock In \emph{Proceedings of the Thirty-Eighth Annual ACM Symposium on
  Theory of Computing (STOC)}, pages 681--690, New York, NY, USA, 2006.
  Association for Computing Machinery.
\newblock ISBN 1595931341.
\newblock \doi{10.1145/1132516.1132612}.

\bibitem[Calderbank et~al.(1997)Calderbank, Rains, Shor, and
  Sloane]{qec_binary_orthogonal_geometry}
A.~R. Calderbank, E.~M. Rains, P.~W. Shor, and N.~J.~A. Sloane.
\newblock Quantum error correction and orthogonal geometry.
\newblock \emph{Phys. Rev. Lett.}, 78:\penalty0 405--408, Jan 1997.
\newblock \doi{10.1103/PhysRevLett.78.405}.

\bibitem[{Gottesman}(1997)]{gottesman_phd}
Daniel {Gottesman}.
\newblock \emph{Stabilizer codes and quantum error correction}.
\newblock PhD thesis, California Institute of Technology, Jan 1997.

\bibitem[Nielsen and Chuang(2011)]{nielsen_chuang}
Michael~A. Nielsen and Isaac~L. Chuang.
\newblock \emph{Quantum Computation and Quantum Information: 10th Anniversary
  Edition}.
\newblock Cambridge University Press, New York, NY, USA, 2011.
\newblock ISBN 978-1107002173.

\bibitem[Hein et~al.(2004)Hein, Eisert, and Briegel]{eisert2004_graph_lu}
M.~Hein, J.~Eisert, and H.~J. Briegel.
\newblock Multiparty entanglement in graph states.
\newblock \emph{Phys. Rev. A}, 69:\penalty0 062311, Jun 2004.
\newblock \doi{10.1103/PhysRevA.69.062311}.

\bibitem[Bouchet(1993)]{bouchet2002_lu_graphs}
Andr{\'e} Bouchet.
\newblock Recognizing locally equivalent graphs.
\newblock \emph{Discrete Mathematics}, 114\penalty0 (1):\penalty0 75 -- 86,
  1993.
\newblock ISSN 0012-365X.
\newblock \doi{10.1016/0012-365X(93)90357-Y}.

\bibitem[Jiang et~al.(2020)Jiang, Sun, Teng, Wu, Wu, and
  Zhang]{2019optimal_cnot_space_depth_cas_pku}
Jiaqing Jiang, Xiaoming Sun, Shang-Hua Teng, Bujiao Wu, Kewen Wu, and Jialin
  Zhang.
\newblock Optimal space-depth trade-off of {CNOT} circuits in quantum logic
  synthesis.
\newblock In \emph{Proceedings of the 2020 ACM-SIAM Symposium on Discrete
  Algorithms (SODA)}, pages 213--229, 2020.
\newblock \doi{10.1137/1.9781611975994.13}.

\bibitem[Moore and Nilsson(2002)]{1998moore_nilsson_cnot_depth}
Cristopher Moore and Martin Nilsson.
\newblock Parallel quantum computation and quantum codes.
\newblock \emph{SIAM J. Comput.}, 31\penalty0 (3):\penalty0 799–815, March
  2002.
\newblock ISSN 0097-5397.
\newblock \doi{10.1137/S0097539799355053}.

\bibitem[Dennis et~al.(2002)Dennis, Kitaev, Landahl, and
  Preskill]{kitaev_qec_measure}
Eric Dennis, Alexei Kitaev, Andrew Landahl, and John Preskill.
\newblock Topological quantum memory.
\newblock \emph{Journal of Mathematical Physics}, 43\penalty0 (9):\penalty0
  4452--4505, 2002.
\newblock \doi{10.1063/1.1499754}.

\bibitem[McClean et~al.(2020)McClean, Rubin, Sung, Kivlichan, Bonet-Monroig,
  Cao, Dai, Fried, Gidney, Gimby, Gokhale, H{\"a}ner, Hardikar,
  Havl{\'{\i}}{\v{c}}ek, Higgott, Huang, Izaac, Jiang, Liu, McArdle, Neeley,
  O'Brien, O'Gorman, Ozfidan, Radin, Romero, Sawaya, Senjean, Setia, Sim,
  Steiger, Steudtner, Sun, Sun, Wang, Zhang, and Babbush]{openfermion}
Jarrod~R McClean, Nicholas~C Rubin, Kevin~J Sung, Ian~D Kivlichan, Xavier
  Bonet-Monroig, Yudong Cao, Chengyu Dai, E~Schuyler Fried, Craig Gidney,
  Brendan Gimby, Pranav Gokhale, Thomas H{\"a}ner, Tarini Hardikar,
  Vojt{\v{e}}ch Havl{\'{\i}}{\v{c}}ek, Oscar Higgott, Cupjin Huang, Josh Izaac,
  Zhang Jiang, Xinle Liu, Sam McArdle, Matthew Neeley, Thomas O'Brien, Bryan
  O'Gorman, Isil Ozfidan, Maxwell~D Radin, Jhonathan Romero, Nicolas P~D
  Sawaya, Bruno Senjean, Kanav Setia, Sukin Sim, Damian~S Steiger, Mark
  Steudtner, Qiming Sun, Wei Sun, Daochen Wang, Fang Zhang, and Ryan Babbush.
\newblock {OpenFermion}: the electronic structure package for quantum
  computers.
\newblock \emph{Quantum Science and Technology}, 5\penalty0 (3):\penalty0
  034014, June 2020.
\newblock \doi{10.1088/2058-9565/ab8ebc}.

\bibitem[Bravyi and Kitaev(2002)]{bravyi_kitaev2002}
Sergey~B. Bravyi and Alexei~Yu. Kitaev.
\newblock Fermionic quantum computation.
\newblock \emph{Annals of Physics}, 298\penalty0 (1):\penalty0 210 -- 226,
  2002.
\newblock ISSN 0003-4916.
\newblock \doi{10.1006/aphy.2002.6254}.

\bibitem[{Bravyi} et~al.(2017){Bravyi}, {Gambetta}, {Mezzacapo}, and
  {Temme}]{bksc_fermion}
Sergey {Bravyi}, Jay~M. {Gambetta}, Antonio {Mezzacapo}, and Kristan {Temme}.
\newblock Tapering off qubits to simulate fermionic {H}amiltonians.
\newblock \emph{arXiv e-prints}, Jan 2017.

\bibitem[Hagberg et~al.(2008)Hagberg, Schult, and
  Swart]{hagberg_exploring_2008}
Aric~A. Hagberg, Daniel~A. Schult, and Pieter~J. Swart.
\newblock Exploring network structure, dynamics, and function using {NetworkX}.
\newblock In Gaël Varoquaux, Travis Vaught, and Jarrod Millman, editors,
  \emph{Proceedings of the 7th Python in Science Conference}, pages 11 -- 15,
  2008.

\bibitem[Arlazarov et~al.(1970)Arlazarov, Dinic, Kronod, and
  Faradez]{four_russians}
V.L. Arlazarov, E.A. Dinic, M.A. Kronod, and I.A. Faradez.
\newblock On economical construction of the transitive closure of an oriented
  graph.
\newblock \emph{Soviet Mathematics Doklady}, pages 1209--10, 1970.

\bibitem[Parrish et~al.(2017)Parrish, Burns, Smith, Simmonett, DePrince,
  Hohenstein, Bozkaya, Sokolov, Di~Remigio, Richard, Gonthier, James,
  McAlexander, Kumar, Saitow, Wang, Pritchard, Verma, Schaefer, Patkowski,
  King, Valeev, Evangelista, Turney, Crawford, and Sherrill]{parrish_psi4_2017}
Robert~M. Parrish, Lori~A. Burns, Daniel G.~A. Smith, Andrew~C. Simmonett,
  A.~Eugene DePrince, Edward~G. Hohenstein, U{\u g}ur Bozkaya, Alexander~Yu.
  Sokolov, Roberto Di~Remigio, Ryan~M. Richard, J{\'e}r{\^o}me~F. Gonthier,
  Andrew~M. James, Harley~R. McAlexander, Ashutosh Kumar, Masaaki Saitow, Xiao
  Wang, Benjamin~P. Pritchard, Prakash Verma, Henry~F. Schaefer, Konrad
  Patkowski, Rollin~A. King, Edward~F. Valeev, Francesco~A. Evangelista,
  Justin~M. Turney, T.~Daniel Crawford, and C.~David Sherrill.
\newblock Psi4 1.1: An open-source electronic structure program emphasizing
  automation, advanced libraries, and interoperability.
\newblock \emph{Journal of Chemical Theory and Computation}, 13\penalty0
  (7):\penalty0 3185--3197, 07 2017.
\newblock \doi{10.1021/acs.jctc.7b00174}.

\end{thebibliography}

\clearpage
\appendix
\section{Example to demonstrate that combining two collections into one never reduces \texorpdfstring{$R$}{R}}\label{app:groupcomb}

Consider the example in Refs~\cite{mcclean2016theory, gokhale2019_commute} where we consider measuring the energy, given by the Hamiltonian
\begin{equation}
    H = -XX - YY + ZZ + IZ + ZI, 
\end{equation} 
of the state $\ket{\psi} = \ket{01}$. The covariance matrix is 
\begin{equation}
    C = \begin{pmatrix}
        \\[1pt]
        1 & 1 & 0 & \textbf{0} & \textbf{0} 
        \\[1pt]
        1 & 1 & 0 & \textbf{0} & \textbf{0} 
        \\[1pt]
        0 & 0 & 0 & 0 & 0 
        \\[1pt]
        \textbf{0} & \textbf{0} & 0 & 0 & 0 
        \\[1pt]
        \textbf{0} & \textbf{0} & 0 & 0 & 0
        \\[1pt]
    \end{pmatrix}.
\end{equation}
The covariance between the non-commuting operators in the upper right and lower left blocks is not defined. We have set them to equal zero for convenience (highlighted in bold). 

First, we consider collecting the Paulis into 
\begin{equation}
    \{-XX, -YY, ZZ\}, \  \{IZ,ZI\}.
\end{equation}
For these collections of Paulis, the coefficient vectors are
\begin{equation}
    \mathbf{a} = \begin{pmatrix}
    \\[1pt]
    -1 
    \\[1pt]
    -1 
    \\[1pt]
    1 
    \\[1pt]
    0 
    \\[1pt]
    0 
    \\[1pt]
    \end{pmatrix},\  
    \mathbf{c} = \begin{pmatrix}
    \\[1pt]
    0 
    \\[1pt]
    0 
    \\[1pt]
    0 
    \\[1pt]
    1 
    \\[1pt]
    1 
    \\[1pt]
    \end{pmatrix}.
\end{equation}
The number of measurements to achieve an accuracy of $\epsilon$ is 
\begin{equation}
\begin{aligned}
    M_g &= \frac{1}{\epsilon^{2}}\left (\sqrt{\mathbf{a}\trans C \mathbf{a}} +  \sqrt{\mathbf{c}\trans C \mathbf{c}}\right )^2 \\
    &= \frac{1}{\epsilon^{2}}\left(\sqrt{4} +  \sqrt{0}\right)^2 \\
    &= \frac{4}{\epsilon^{2}}.
\end{aligned}
\end{equation}
Now, let us consider breaking up the first collection into 
\begin{equation}
    \{-XX\}, \ \{-YY, ZZ\}.
\end{equation}
In this case, the coefficient vectors are
\begin{equation}
\mathbf{a} = \begin{pmatrix}
    \\[1pt]
    -1 
    \\[1pt]
    0 
    \\[1pt]
    0 
    \\[1pt]
    0 
    \\[1pt]
    0 
    \\[1pt]
    \end{pmatrix},\ 
    \mathbf{b} = \begin{pmatrix}
    \\[1pt]
    0 
    \\[1pt]
    -1 
    \\[1pt]
    1 
    \\[1pt]
    0 
    \\[1pt]
    0 
    \\[1pt]
    \end{pmatrix}, \ 
    \mathbf{c} = \begin{pmatrix}
    \\[1pt]
    0 
    \\[1pt]
    0 
    \\[1pt]
    0 
    \\[1pt]
    1 
    \\[1pt]
    1 
    \\[1pt]
    \end{pmatrix}.
\end{equation}
The number of measurements required to attain an accuracy $\epsilon$ is therefore
\begin{equation}
\begin{split}
    M_g &= \frac{1}{\epsilon^{2}}\left(\sqrt{\mathbf{a}\trans C \mathbf{a}} + \sqrt{\mathbf{b}\trans C \mathbf{b}} +  \sqrt{\mathbf{c}\trans C \mathbf{c}}\right)^2  \\
    &= \frac{1}{\epsilon^{2}}\left(\sqrt{1} + \sqrt{1} + \sqrt{0}\right)^2 \\
    &= \frac{4}{\epsilon ^2}.
    \end{split}
\end{equation}
Therefore, under the optimal measurement strategy it is not preferable to break the $\{-XX, -YY, ZZ\}$ collection into $\{-XX\}$ and $\{-YY, ZZ\}$. In this specific example we have equality because for $\alpha=-\beta$ we have $\langle{\alpha \mathbf{a} + \beta \mathbf{b}},{\alpha \mathbf{a} + \beta \mathbf{b}}\rangle = 0$. 

\section{Derivation of \texorpdfstring{$\hat{R}$}{R} formula}\label{app:rhatder}
\begin{claim}
    For $R$ as defined in Eq.~\eqref{eq:rratio}, if all variances and covariances are replaced with their expectation value over uniform spherical distribution, we obtain a new metric, $\hat{R}$, given by
    \begin{equation}\label{eq:Rtilde}
    \hat{R} = \left(\frac{\sum_{i=1}^{N}\sum_{j=1}^{N_i} |a_{ij}|}{\sum_{i=1}^{N} \sqrt{\sum_{j=1}^{N_i}a_{ij}^2}
    }\right)^2.
\end{equation}
\end{claim}
\begin{proof}
The variance of a single Pauli operator is
\begin{equation}
    \Var[P_i] \coloneqq  1 - \expval{P_i}^2. 
\end{equation}
The expectation of this variance for all $P_i \neq I$ is 
\begin{equation}
\begin{split}
    \mathbb{E}[\Var P_i] &= 1 - \mathbb{E}[\expval{P_i}^2] \\
    &= 1- \int [\bra{\psi} P_i \ket{\psi}]^2 \diff \psi \\ 
    &= 1 - \alpha_{n},
    \end{split}
\end{equation}
where $\alpha_{n} \coloneqq 1/(2^{n}+1)$, with $n$ the number of qubits, is independent of $P_{i}$~\cite[Exercise~7.3]{watrous_2018}. Trivially, $\Var [I] = 0$. In addition, it was shown in Ref.~\cite{gokhale2019_commute} that
\begin{equation}
    \mathbb{E}[\covs{P_i}{P_j}] = 0,
\end{equation}
for all $P_i \neq P_j$. Substitution of these results yields Eq.~\eqref{eq:Rtilde}.
\end{proof} 

\section{Binary representation}\label{app:binrep}
The Pauli subset $\mathcal{P}_{n}$ on $n$-qubits is a collection of $4^{n+1}$ elements defined by
\begin{align}\label{eq:pauli_group}
    \mathcal{P}_{n} = \{&i^{k}  \ \sigma_{1}\otimes \dots \otimes \sigma_{n} \ | \ \sigma_{i} \in \{I,X,Y,Z\}, \nonumber \\
        &\ k \in \{0,1,2,3\}  \}.
\end{align}

The binary representation, first introduced by~\citeauthor*{qec_binary_orthogonal_geometry}~\cite{qec_binary_orthogonal_geometry},  is a representation of $\mathcal{P}_{n}$ as binary vectors. In this representation, Paulis differing only in phase $i^{k}$ are represented in the same way.

Single-qubit Paulis are represented by 2-dimensional binary vectors, so that
\begin{equation}
\begin{alignedat}{3}
    &\sigma_{00}  \coloneqq \ && I    &&\mapsto (0, 0), \\
    &\sigma_{01} \coloneqq \ && X  &&\mapsto (0, 1),\\
    &\sigma_{10}  \coloneqq \ && Z &&\mapsto (1, 0),\\
    &\sigma_{11}  \coloneqq \ && Y &&\mapsto (1, 1).
\end{alignedat}
\end{equation}
An $n$-qubit Pauli
\begin{equation}
    \sigma_{u_{1}v_{1}} \otimes \dots \otimes \sigma_{u_{n}v_{n}}
\end{equation}
is then represented by the $2n$-dimensional binary vector
\begin{equation}
(u_{1}, \dots, u_{n}, v_{1}, \dots, v_{n})\trans.
\end{equation}
In this representation, two $n$-qubit Paulis with binary vectors $a$ and $b$ commute if and only if
\begin{equation}\label{eq:appendix_Jvector}
    a\trans J_{2n} b = 0,
\end{equation} 
where $J_{2n}$ denotes the $2n\times 2n$ matrix
\begin{equation}\label{eq:appendix_J}
    J_{2n} \coloneqq
    \begin{pmatrix}
        0 & I_{n} \\
        \\[1pt]
        I_{n} & 0
    \end{pmatrix}.
\end{equation}

Given a set $\mathcal{S}$ of $m$ $n$-qubit Paulis, we can write down a corresponding $2n \times m$ binary matrix $S$ where each column represents a Pauli. Then, from Eq.~\eqref{eq:appendix_Jvector}, we deduce that all Paulis in $\mathcal{S}$ mutually commute if and only if
\begin{equation}
    S\trans J_{2n} S = 0_{m},
\end{equation}
which recovers Eq.~\eqref{eq:commutativitycheck} in the main text. We say that the set $\mathcal{S}$ of Paulis is independent if the matrix $S$ has rank $m$.

We shall often find it helpful to write $S$ in terms of its upper half $S^{(Z)}$ and lower half $S^{(X)}$, separated by a horizontal line for visual-aid, i.e.,
\begin{equation}
    S = \begin{pmatrix}
    \\[1pt]
    S^{(Z)}
    \\[1pt]
    \hline
    \\[1pt]
    S^{(X)}
    \\[1pt]
    \end{pmatrix}.
\end{equation}

The conjugation action of quantum gates on $\mathcal{S}$ can be represented as transformations to the matrix $S$. For example, we document the transformations on $S$ that represent four common quantum gates. In the following, addition is $\mathrm{mod}\ 2$ and $p$ ranges over all columns $\{1,\dots, m\}$.

\begin{itemize}\renewcommand{\labelitemi}{\tiny$\blacksquare$}
    \item \textsc{CZ} on qubits $i$ and $j$:\\  $ S_{ip} \leftarrow S_{ip} + S_{j+n, p}$,\\ 
        $S_{jp} \leftarrow S_{jp} + S_{i+n,p}$.
    \item \textsc{CNOT} on control-qubit $i$ and target-qubit $j$: \\
        $ S_{ip} \leftarrow S_{ip} + S_{jp}$,\\ 
        $S_{j+n,p} \leftarrow S_{j+n,p} + S_{i+n,p}$.
    \item \textsc{Hadamard} (\textsc{H}) on qubit $i$:\\
        \hspace{\parindent} $S_{ip} \leftrightarrow S_{i+n,p}$.
    \item \textsc{Phase} (\textsc{P}) on qubit $i$: \\
        \hspace{\parindent} $S_{ip} \leftarrow S_{ip} + S_{i+n,p}$.
\end{itemize}

These rules can be directly verified by conjugating $X_{i}, \ X_{j}, \ Z_{i}, \ Z_{j}$ by the listed gates. They are also reproduced in Sec.~\ref{sec:rotation_circuit_constructions} of the main text. 

\section{\textsc{CZ}-construction example}\label{app:czconstruction}
We walk through our \textsc{CZ}-construction for a specific example. In this example, we would like to obtain measurements simultaneously of a collection $\mathcal{S}'_\text{start}$ of six four-qubit Paulis given by
\begin{equation}
\begin{aligned}
    P_{1} &= Z_{1}Z_{2}Z_{3}Z_{4}, \\
    P_{2} &= X_{1}X_{2}Y_{3}Y_{4}, \\
    P_{3} &= Y_{1}Y_{2}X_{3}X_{4}, \\
    P_{4} &=
    Y_{2}X_{3}, \\
    P_{5} &=
    Y_{1}X_{4}, \\
    P_{6} &= X_{1}Z_{2}Z_{3}Y_{4}.
\end{aligned}
\end{equation}
We can represent these Paulis in a matrix $\Sstart'$ with
\begin{equation}
    \Sstart' = 
    \begin{pmatrix}
        \\[1pt]
        1 & 0 & 1 & 0 & 1 & 0
        \\[1pt]
        1 & 0 & 1 & 1 & 0 & 1
        \\[1pt]
        1 & 1 & 0 & 0 & 0 & 1
        \\[1pt]
        1 & 1 & 0 & 0 & 0 & 1
        \\[1pt]
        \hline
        \\[1pt]
        0 & 1 & 1 & 0 & 1 & 1
        \\[1pt]
        0 & 1 & 1 & 1 & 0 & 0
        \\[1pt]
        0 & 1 & 1 & 1 & 0 & 0
        \\[1pt]
        0 & 1 & 1 & 0 & 1 & 1
        \\[1pt]
    \end{pmatrix}.
\end{equation}
By Gaussian elimination, we find the reduced row echelon form of $\Sstart'$ to be
\begin{equation}
    \begin{pmatrix}
        \\[1pt]
        1 & 0 & 1 & 0 & 1 & 0
        \\[1pt]
        0 & 1 & 1 & 0 & 1 & 1
        \\[1pt]
        0 & 0 & 0 & 1 & 1 & 1
        \\[1pt]
        0 & 0 & 0 & 0 & 0 & 0
        \\[1pt]
        \hline
        \\[1pt]
        0 & 0 & 0 & 0 & 0 & 0
        \\[1pt]
        0 & 0 & 0 & 0 & 0 & 0
        \\[1pt]
        0 & 0 & 0 & 0 & 0 & 0
        \\[1pt]
        0 & 0 & 0 & 0 & 0 & 0
        \\[1pt]
    \end{pmatrix}.
\end{equation}
The pivot columns are numbers 1, 2 and 4 which tells us that $P_{1}$, $P_{2}$ and $P_{4}$ are the three independent Paulis from which the remaining Paulis in $\mathcal{S}'_\text{start}$ can be constructed. Therefore, we can write $\Sstart' = \Sstart R_{0}^{-1}$, where
\begin{equation}
     \Sstart \coloneqq 
    \begin{pmatrix}
        \\[1pt]
        1 & 0 & 0
        \\[1pt]
        1 & 0 & 1
        \\[1pt]
        1 & 1 & 0 
        \\[1pt]
        1 & 1 & 0 
        \\[1pt]
        \hline
        \\[1pt]
        0 & 1 & 0 
        \\[1pt]
        0 & 1 & 1 
        \\[1pt]
        0 & 1 & 1
        \\[1pt]
        0 & 1 & 0
        \\[1pt]
    \end{pmatrix}, \  
    R_{0}^{-1} \coloneqq
    \begin{pmatrix}
        \\[1pt]
        1 & 0 & 1 & 0 & 1 & 0
        \\[1pt]
        0 & 1 & 1 & 0 & 1 & 1
        \\[1pt]
        0 & 0 & 0 & 1 & 1 & 1
        \\[1pt]
    \end{pmatrix}.
\end{equation}

Note that the inverse on $R_{0}^{-1}$ is purely notational. Now, the lower half $\Sstart^{(X)}$ of $\Sstart$ has column echelon form
\begin{equation}
    \Sstart^{(X)} = 
    \begin{pmatrix}
        \\[1pt]
        1 & 0 & 0
        \\[1pt]
        1 & 1 & 0
        \\[1pt]
        1 & 1 & 0
        \\[1pt]
        1 & 0 & 0
        \\[1pt]
    \end{pmatrix},
\end{equation}
and so the first two rows are pivot rows. In order to give the lower half of $\Sstart$ a rank of $k=3$, we therefore apply a \textsc{Hadamard} to the rows corresponding to qubits 3 and 4 so that
\begin{equation}
    S_{1} \coloneqq Q_{1} \Sstart =     \begin{pmatrix}
        \\[1pt]
        1 & 0 & 0
        \\[1pt]
        1 & 0 & 1
        \\[1pt]
        0 & 1 & 1
        \\[1pt]
        0 & 1 & 0
        \\[1pt]
        \hline
        \\[1pt]
        0 & 1 & 0
        \\[1pt]
        0 & 1 & 1
        \\[1pt]
        1 & 1 & 0
        \\[1pt]
        1 & 1 & 0
        \\[1pt]
    \end{pmatrix}, 
\end{equation}
where 
\begin{equation}
Q_{1} \coloneqq
\begin{pmatrix}[cccc|cccc]
    \\[1pt]
    1 & 0 & 0 & 0 & 0 & 0 & 0 & 0
    \\[1pt]
    0 & 1 & 0 & 0 & 0 & 0 & 0 & 0
    \\[1pt]
    0 & 0 & 0 & 0 & 0 & 0 & 1 & 0
    \\[1pt]
    0 & 0 & 0 & 0 & 0 & 0 & 0 & 1
    \\[1pt]
    \hline
    \\[1pt]
    0 & 0 & 0 & 0 & 1 & 0 & 0 & 0
    \\[1pt]
    0 & 0 & 0 & 0 & 0 & 1 & 0 & 0
    \\[1pt]
    0 & 0 & 1 & 0 & 0 & 0 & 0 & 0
    \\[1pt]
    0 & 0 & 0 & 1 & 0 & 0 & 0 & 0
    \\[1pt]
\end{pmatrix}. 
\end{equation}
    
The lower half of $S_{1}$ now has rank 3, and performing Gaussian elimination on it, we find
\begin{equation}
    S_{2} \coloneqq Q_{1} \Sstart R_{1} = 
    \begin{pmatrix}
        \\[1pt]
        1 & 0 & 1
        \\[1pt]
        0 & 1 & 1
        \\[1pt]
        0 & 1 & 0
        \\[1pt]
        1 & 0 & 0
        \\[1pt]
        \hline
        \\[1pt]
        1 & 0 & 0
        \\[1pt]
        0 & 1 & 0
        \\[1pt]
        0 & 0 & 1
        \\[1pt]
        0 & 0 & 1
        \\[1pt]
    \end{pmatrix},
\end{equation}
where
\begin{equation}
    R_{1} = 
    \begin{pmatrix}
        \\[1pt]
        0 & 1 & 0
        \\[1pt]
        0 & 1 & 1
        \\[1pt]
        1 & 1 & 0
        \\[1pt]
    \end{pmatrix}^{-1} = 
    \begin{pmatrix}
        \\[1pt]
        1 & 0 & 1
        \\[1pt]
        1 & 0 & 0
        \\[1pt]
        1 & 1 & 0
        \\[1pt]
    \end{pmatrix}.
\end{equation}

We now extend $S_{2}$ to a rank $n=4$ matrix by adding a column that corresponds to a fourth Pauli $P_{ext}$. In the main text, this is the crucial basis extension step from $S_{2}\rightarrow S_{3}$ which might have seemed fortuitous. In fact, the extension was systematically obtained as follows.

To make our reasoning clearer, let us represent  $S_{2}$ alternatively by the matrix 
\begin{equation}\label{eq:appendix_pauliS2}
P(S_{2})\coloneqq \begin{pmatrix}
    \\[1pt]
    Y & I & I & Z
    \\[1pt]
    I & Y & Z & I
    \\[1pt]
    Z & Z & X & X
    \\[1pt]
\end{pmatrix},
\end{equation}
where each row corresponds to a Pauli operator given by a column of $S_{2}$. Looking at the form of $P(S_{2})$, we see that we can place $X$ in the $4$th qubit position of $P_{ext}$ (and nowhere else) to ensure $P_{ext}$ is independent of the other Paulis. Then we observe that the left 3-by-3 sub-matrix of $P(S_{2})$ has $X/Y$ on the diagonal and $I/Z$ everywhere else. This means we can place $I/Z$ in the other qubit positions of $P_{ext}$ depending on whether the $X$ in its $4$th qubit position commutes with the $4$th position terms of the other Paulis. 

By this prescription, we find $P_{ext} = Z_{1}I_{2}I_{3}X_{4}$. Therefore, $S_{2}$ is extended to
\begin{equation}
    S_{3} \coloneqq 
    \begin{pmatrix}
        \\[1pt]
        1 & 0 & 1 & 1
        \\[1pt]
        0 & 1 & 1 & 0
        \\[1pt]
        0 & 1 & 0 & 0
        \\[1pt]
        1 & 0 & 0 & 0
        \\[1pt]
        \hline
        \\[1pt]
        1 & 0 & 0 & 0
        \\[1pt]
        0 & 1 & 0 & 0
        \\[1pt]
        0 & 0 & 1 & 0
        \\[1pt]
        0 & 0 & 1 & 1
        \\[1pt]
    \end{pmatrix},
\end{equation}
and
\begin{equation}
    Q_{1} \Sstart R_{1} = S_{3}R_{2}^{-1},
\end{equation}
where
\begin{equation}
    R_{2}^{-1} \coloneqq \begin{pmatrix}
        \\[1pt]
        1 & 0 & 0
        \\[1pt]
        0 & 1 & 0
        \\[1pt]
        0 & 0 & 1
        \\[1pt]
        0 & 0 & 0
        \\[1pt]
    \end{pmatrix}.
\end{equation}
Note that the inverse on $R_{2}^{-1}$ is also purely notational. The lower half of $S_{3}$ is full-rank and so we can take its inverse to find
\begin{equation}
    R_{3} = 
    \begin{pmatrix}
        \\[1pt]
        1 & 0 & 0 & 0
        \\[1pt]
        0 & 1 & 0 & 0
        \\[1pt]
        0 & 0 & 1 & 0
        \\[1pt]
        0 & 0 & 1 & 1
        \\[1pt]
    \end{pmatrix}^{-1}
    = 
    \begin{pmatrix}
        \\[1pt]
        1 & 0 & 0 & 0
        \\[1pt]
        0 & 1 & 0 & 0
        \\[1pt]
        0 & 0 & 1 & 0
        \\[1pt]
        0 & 0 & 1 & 1
        \\[1pt]
    \end{pmatrix}
\end{equation}
and
\begin{equation}
    S_{3} R_{3} = 
    \begin{pmatrix}
        \\[1pt]
        1 & 0 & 0 & 1
        \\[1pt]
        0 & 1 & 1 & 0
        \\[1pt]
        0 & 1 & 0 & 0
        \\[1pt]
        1 & 0 & 0 & 0
        \\[1pt]
        \hline
        \\[1pt]
        1 & 0 & 0 & 0
        \\[1pt]
        0 & 1 & 0 & 0
        \\[1pt]
        0 & 0 & 1 & 0
        \\[1pt]
        0 & 0 & 0 & 1
        \\[1pt]
    \end{pmatrix}.
\end{equation}
Finally, we apply \textsc{Phase} to qubits 1 and 2 so that
\begin{equation}
    S_{4} = Q_{2} S_{3} R_{3} = 
    \begin{pmatrix}
        \\[1pt]
        0 & 0 & 0 & 1
        \\[1pt]
        0 & 0 & 1 & 0
        \\[1pt]
        0 & 1 & 0 & 0
        \\[1pt]
        1 & 0 & 0 & 0
        \\[1pt]
        \hline
        \\[1pt]
        1 & 0 & 0 & 0
        \\[1pt]
        0 & 1 & 0 & 0
        \\[1pt]
        0 & 0 & 1 & 0
        \\[1pt]
        0 & 0 & 0 & 1
        \\[1pt]
    \end{pmatrix},
\end{equation}
where
\begin{equation}
    Q_{2} = 
    \begin{pmatrix}[cccc|cccc]
        \\[1pt]
        1 & 0 & 0 & 0 & 1 & 0 & 0 & 0
        \\[1pt]
        0 & 1 & 0 & 0 & 0 & 1 & 0 & 0
        \\[1pt]
        0 & 0 & 1 & 0 & 0 & 0 & 0 & 0
        \\[1pt]
        0 & 0 & 0 & 1 & 0 & 0 & 0 & 0
        \\[1pt]
        \hline
        \\[1pt]
        0 & 0 & 0 & 0 & 1 & 0 & 0 & 0
        \\[1pt]
        0 & 0 & 0 & 0 & 0 & 1 & 0 & 0
        \\[1pt]
        0 & 0 & 0 & 0 & 0 & 0 & 1 & 0
        \\[1pt]
        0 & 0 & 0 & 0 & 0 & 0 & 0 & 1
        \\[1pt]
    \end{pmatrix}.
\end{equation}

$S_{4}$ is of the form of a graph state and represents the following Paulis:
\begin{equation}
\begin{aligned}
    \tilde{P}_{1} &= X_{1}\, I_{2}\, I_{3}\, Z_{4},\\
    \tilde{P}_{2} &= I_{1}\, X_{2}\, Z_{3}\, I_{4},\\
    \tilde{P}_{3} &= I_{1}\, Z_{2}\, X_{3}\, I_{4},\\
    \tilde{P}_{4} &= Z_{1}\, I_{2}\, I_{3}\, X_{4}.
\end{aligned}
\end{equation}
We now have
\begin{equation}
    \Sstart' = Q^{-1} S_{4} R^{-1},
\end{equation}
where
\begin{equation}
    Q^{-1} \coloneqq Q_{1}^{-1} Q_{2}^{-1} = 
    \begin{pmatrix}[cccc|cccc]
        \\[1pt]
        1 & 0 & 0 & 0 & 1 & 0 & 0 & 0
        \\[1pt]
        0 & 1 & 0 & 0 & 0 & 1 & 0 & 0
        \\[1pt]
        0 & 0 & 0 & 0 & 0 & 0 & 1 & 0
        \\[1pt]
        0 & 0 & 0 & 0 & 0 & 0 & 0 & 1
        \\[1pt]
        \hline
        \\[1pt]
        0 & 0 & 0 & 0 & 1 & 0 & 0 & 0
        \\[1pt]
        0 & 0 & 0 & 0 & 0 & 1 & 0 & 0
        \\[1pt]
        0 & 0 & 1 & 0 & 0 & 0 & 0 & 0
        \\[1pt]
        0 & 0 & 0 & 1 & 0 & 0 & 0 & 0
        \\[1pt]
    \end{pmatrix}
\end{equation}
and
\begin{equation}
    R^{-1} \coloneqq R_{3}^{-1} R_{2}^{-1} R_{1}^{-1} R_{0}^{-1} = 
    \begin{pmatrix}
        \\[1pt]
        0 & 1 & 1 & 0 & 1 & 1
        \\[1pt]
        0 & 1 & 1 & 1 & 0 & 0
        \\[1pt]
        1 & 1 & 0 & 0 & 0 & 1
        \\[1pt]
        1 & 1 & 0 & 0 & 0 & 1 
        \\[1pt]
    \end{pmatrix}.
\end{equation}

The rotation circuit is shown in Fig.~\ref{fig:examplecirc}. Using $Q^{-1}$, $S_{4}$ and $R^{-1}$, we can work out that the phases for the six original operators are $\left ( \begin{smallmatrix} +1 & +1 & +1 & -1 & -1 & -1 \end{smallmatrix} \right )$. Therefore, we can construct measurements of the original Pauli strings as follows:
\begin{itemize}
\renewcommand{\labelitemi}{\tiny$\blacksquare$}
    \item $P_{1}$ from product of measurements of qubits 3 and 4,
    \item $P_{2}$ from product of measurements of qubits 1 to 4,
    \item $P_{3}$ from product of measurements of qubits 1 and 2,
    \item $P_{4}$ from the negative of measurement of qubit 2,
    \item $P_{5}$ from the negative of measurement of qubit 1,
    \item $P_{6}$ from the negative product of measurements of qubits 1, 3 and 4.
\end{itemize}

\begin{figure}
\centering
\mbox{ 
\Qcircuit @C=1.5em @R=1.4em {\lstick{q1}  & \gate{S} & \ctrl{3}  & \qw & \gate{H} & \meter  & \cw \\
\lstick{q2} & \gate{S} & \qw & \ctrl{1} & \gate{H} & \meter  & \cw\\
\lstick{q3} & \gate{H} & \qw & \control \qw & \gate{H} & \meter  & \cw\\
\lstick{q4} & \gate{H} & \control \qw & \qw & \gate{H} & \meter  & \cw\\
}}
\centering{}\caption{{{}The rotation circuit $U$ for the \textsc{CZ}-construction walk-through example.}}
\label{fig:examplecirc}
\end{figure}
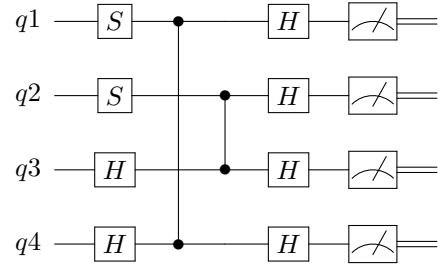

\section{Proof of \texorpdfstring{$O(kn/\log k)$}{O(kn/log k)}}\label{app:fourrussians}
We prove the following Claim~\ref{claim: scaling} via arguments of \citeauthor*{patel_markov_hayes_cnot}~\cite{patel_markov_hayes_cnot}. As acknowledged in Ref.~\cite{patel_markov_hayes_cnot}, these arguments originate from the ``Method of Four Russians''~\cite{four_russians}. Note that row operations correspond to \textsc{CNOT} gates, as explained in detail in Ref.~\cite{patel_markov_hayes_cnot}.

\begin{claim}\label{claim: scaling}
    Let $M$ be a $n\times n$ matrix with block form
\begin{equation}
    \begin{pmatrix}
        I_{k} & A
        \\[1pt]
        0 & I_{n-k}
    \end{pmatrix},
    \end{equation}
    where $A$ is any $k\times (n-k)$ matrix. Then $O(kn/\log k)$ row operations suffice to reduce $M$ to identity $I_{n}$.
\end{claim}

\begin{proof}
    Let $m$ be a constant we later choose. Partition $A$ into $l \coloneqq (n-k)/m$ consecutive column-blocks $A_{i}$, each containing $m$ columns. 
    
    Start at $A_{1}$ and eliminate any duplicate rows using at most $k$ row operations. There then remain at most $2^{m}$ unique rows in $A_{1}$ which can be eliminated by at most $m2^{m-1}$ row operations that add rows from $I_{n-k}$. $A_{1}$ is now zero. 
    
    For each of $A_{2}, \dots, A_{l}$ perform the same operation as was done to $A_{1}$. $M$ then becomes
    \begin{equation}
        \begin{pmatrix}
        B & 0
        \\[1pt]
        0 & I_{n-k}
    \end{pmatrix},
    \end{equation}
    where $B$ is some $k\times k$ matrix that must be invertible. $B$ can then be row-reduced to $I_{k}$ using $O(k^{2}/\log k)$ by the result of Ref.~\cite{patel_markov_hayes_cnot}.
    
    The total number $N$ of row operations is therefore
    \begin{equation}
        N = \left(k+m2^{m-1}\right)\ \frac{n-k}{m} + O\left (\frac{k^{2}}{\log k}\right).
    \end{equation}

    Choosing $m = \alpha \log k$, we find
    \begin{equation}
        N = \frac{k(n-k)}{\alpha \log k} + \frac{k^{\alpha}(n-k)}{2} +  O\left(\frac{k^{2}}{\log k}\right),
    \end{equation}
    which is $O(kn/\log k)$ provided $\alpha <1$.
\end{proof}

\section{NP-hardness of maximising \texorpdfstring{$\hat{R}$}{R}}\label{app:rhat_nphard}
In this Appendix, we show that maximising $\hat{R}$ is NP-hard by a reduction argument.

Let $G = (V,E)$ be a simple graph with vertices $V\coloneqq\{1,2,\dots,n\}$ and edges $E$. Following Ref.~\cite[Appendix~A]{gokhale2019_commute}, we then define the operator $O = \sum_{i=1}^n P_i$ on $n$-qubits, where $P_i = \otimes _{j=1}^n P_i^{(j)}$ is a tensor product of $n$ single-qubit Paulis, so that
\begin{equation}
P_i^{(j)} \coloneqq \begin{cases}
   Z \quad \text{if $i=j$,}\\[2\jot]
   X \quad \text{if $j>i$ and $(i,j)\not\in E$,}\\[2\jot]
   I \quad \text{otherwise}.
\end{cases}
\end{equation}
It can be easily seen~\cite[Appendix~A]{gokhale2019_commute} that a commuting collection of the $P_i$s corresponds exactly to a clique in $G$.

For an arrangement $\mathcal{A}$ of $O$ into $N$ commuting collections with sizes $m_i\geq 1$, we consider the term appearing in
the denominator of $\hat{R}$, and see that
\begin{equation}
    D(\mathcal{A}) \coloneqq \sum_{i=1}^{N} \sqrt{\sum_{j=1}^{m_i}|a_{ij}|^2} = \sum_{i=1}^{N} \sqrt{m_i},
\end{equation}
where the second equality is because the coefficients of $P_i$s in $O$ are all $1$s and we do not consider splitting these coefficients. We have
\begin{equation}\label{eq:cauchy}
    N\leq D(\mathcal{A}) \leq \sqrt{n}\ \sqrt{N} \leq \sqrt{n} \ N,
\end{equation}
where the second inequality follows from Cauchy-Schwarz and $\sum_{i=1}^N m_i = n$.

Now, let $N_0$ denote the minimum number of commuting collections corresponding to some arrangement $\mathcal{A}_0$ and let $N_1$ denote the number of commuting collections in an arrangement $\mathcal{A}_1$ that maximises $\hat{R}$. By definition, $N_0\leq N_1$. Since maximising $\hat{R}$ is equivalent to minimising $D(\mathcal{A})$, we also have
\begin{equation}\label{eq:np-hard_approx}
    N_1 \leq D(\mathcal{A}_1) \leq  D(\mathcal{A}_0) \leq \sqrt{n} \ N_0,
\end{equation}
using Eq.~\eqref{eq:cauchy} for the first and third inequalities. Hence, $D(\mathcal{A}_1)$ is a $\sqrt{n}$-approximation to $N_0$. But $N_0$ is also the solution to the \textsc{Min-Clique-Cover} problem on $G$ by construction. So maximising $\hat{R}$ involves computing a $\sqrt{n}$-approximation to \textsc{Min-Clique-Cover} which is NP-hard by Ref.~\cite[Theorem 1.2]{zuckerman_2006}. 

In fact, Ref.~\cite[Theorem 1.2]{zuckerman_2006} gives the stronger result that computing a $n^{1-\epsilon}$ approximation to \textsc{Min-Clique-Cover} is NP-hard for any $\epsilon>0$. This implies that it is NP-hard to compute a $n^{1/2-\epsilon}$ approximation to $D(\mathcal{A}_1)$ for any $\epsilon>0$. In other words, it is NP-hard to compute a $1/n^{1/2-\epsilon}$ approximation to the maximum $\hat{R}$ for any $\epsilon>0$.

\onecolumngrid
\newpage
\section{Full numerical results}~\label{app:full_numerical_results}
In this Appendix, we present the full version of Table~\ref{tab:reddat} in the main text.
\begin{table}[hbt!]
\centering
\footnotesize
\renewcommand{\arraystretch}{1.38} %
\resizebox{17cm}{!}{%
\Large
\begin{tabular}{@{}C{1.5cm}R{4.5em}R{4.5em} R{0.3cm} R{3.5em}R{3.5em}R{3.5em} R{0.3cm} R{3.5em}R{3.5em}R{3.5em}R{3.5em} R{0.3cm} R{3.5em}R{3.5em}R{3.5em}
@{}}\toprule

\rule{0pt}{1.5em} \multirow{3}{*}{Molecule} & \hspace{1.3em}
\multirowcell{2}[-0.9em]{\hfill $n$\\ \hfill qubits} & \hspace{1.3em} \multirowcell{2}[-0.7em]{\hfill $t$\\ \hfill Paulis} &  &\multicolumn{3}{c}{Arrangement} & \hspace{1em} &\multicolumn{4}{c}{Ratios $R$, $\hat{R}$} & \hspace{1em} &\multicolumn{3}{c}{Rotation Circuit $2q$-size}\\
\hhline{~~~~---~----~---}\noalign{\vskip-0.5pt}
\rule{0pt}{1.5em} &  &  & \hspace{1em} & $N$ & $ \overline{m_{i}}$ & $\overline{k_{i}}$ & \hspace{1em} & $R$ min & $R$ mean & $R$ max & $\hat{R}$ & \hspace{1em} & max theory & max true & mean \\ 

\hhline{----------------}
H$_{2}$ & 2 & 4 & \hspace{1em}      & 2 & 2.00 & 1.50 & \hspace{1em} & 1.09 & 1.93 & 4.60 & 1.76 & \hspace{1em} & 0 & 0 & 0\\
H$_{3}^{+}$ & 4 & 59 & \hspace{1em} & 10 & 5.90 & 3.50 & \hspace{1em} & 3.76 & 11.92 & 33.04 & 10.25 & \hspace{1em} & 6 & 3 & 0.80\\
LiH & 10 & 630 & \hspace{1em}       & 41 & 15.37 & 6.85 & \hspace{1em} & 19.60 & 24.91 & 34.74 & 23.97 & \hspace{1em} & 45 & 18 & 5.29\\
OH$^{-}$ & 10 & 630 & \hspace{1em}  & 38 & 16.58 & 7.29 & \hspace{1em} & 6.32 & 8.90 & 12.86 & 8.51 & \hspace{1em} & 45 & 17 & 5.63\\
HF & 10 & 630 & \hspace{1em}        & 39 & 16.15 & 6.97 & \hspace{1em} & 6.07 & 8.57 & 12.27 & 8.21 & \hspace{1em} & 45 & 16 & 5.74\\
H$_{2}$O & 12 & 1085 & \hspace{1em} & 51 & 21.27 & 9.04 & \hspace{1em} & 7.68 & 11.27 & 16.96 & 10.67 & \hspace{1em} & 66 & 26 & 7.37\\
BH$_{3}$ & 14 & 1584 & \hspace{1em} & 66 & 24.00 & 10.36 & \hspace{1em} & 17.21 & 20.93 & 32.13 & 20.05 & \hspace{1em} & 91 & 26 & 9.56\\
NH$_{3}$ & 14 & 3608 & \hspace{1em} & 118 & 30.58 & 11.34 & \hspace{1em} & 12.65 & 15.96 & 26.93 & 15.31 & \hspace{1em} & 91 & 28 & 10.26\\
CH$_{4}$ & 16 & 3887 & \hspace{1em} & 123 & 31.60 & 13.39 & \hspace{1em} & 16.96 & 21.63 & 29.33 & 20.27 & \hspace{1em} & 120 & 45 & 16.75\\
O$_{2}$ & 18 & 2238 & \hspace{1em}  & 67 & 33.40 & 13.48 & \hspace{1em} & - & - & - & 20.23 & \hspace{1em} & 153 & 44 & 21.57\\
N$_{2}$ & 18 & 2950 & \hspace{1em}  & 78 & 37.82 & 13.91 & \hspace{1em} & - & - & - & 22.10 & \hspace{1em} & 153 & 53 & 20.42\\
CO & 18 & 4426 & \hspace{1em}       & 128 & 34.58 & 13.48 & \hspace{1em} & - & - & - & 21.31 & \hspace{1em} & 153 & 50 & 20.23\\
HCl & 18 & 4538 & \hspace{1em} & 123 & 36.89 & 13.87 & \hspace{1em} & - & - & - & 10.36 & \hspace{1em} & 153 & 49 & 20.35\\
NaH & 18 & 5850 & \hspace{1em} & 135 & 43.33 & 14.73 & \hspace{1em} & - & - & - & 12.90 &
\hspace{1em} & 153 & 45 & 21.44\\
H$_{2}$S & 20 & 6277 & \hspace{1em} & 147 & 42.70 & 16.06 & \hspace{1em} & - & - & - & 11.60 & \hspace{1em} & 190 & 58 & 25.98\\
PH$_{3}$ & 22 & 19746 & \hspace{1em} & 304 & 64.95 & 18.77 & \hspace{1em} & - & - & - & 13.05 & \hspace{1em} & 231 & 67 & 28.02\\
SiH$_{4}$ & 24 & 18713 & \hspace{1em} & 304 & 61.56 & 20.98 & \hspace{1em} & - & - & - & 13.94 & \hspace{1em} & 276 & 77 & 36.03\\
NaF & 26 & 16538 & \hspace{1em} & 287 & 57.62 & 20.44 & \hspace{1em} & - & - & - & 23.36 & \hspace{1em} & 325 & 90 & 42.28\\
LiCl & 26 & 17044 & \hspace{1em} & 292 & 58.37 & 20.46 & \hspace{1em} & - & - & - & 12.22 & \hspace{1em} & 325 & 89 & 39.42\\
KH & 26 & 24290 & \hspace{1em} & 325 & 74.74 & 22.30 & \hspace{1em} & - & - & - & 12.87 & \hspace{1em} & 325 & 115 & 45.18\\
CO$_{2}$ & 28 & 11429 & \hspace{1em} & 216 & 52.91 & 21.02 & \hspace{1em} & - & - & - & 38.47 & \hspace{1em} & 378 & 104 & 44.51\\
F$_{2}$O & 28 & 20541 & \hspace{1em} & 317 & 64.80 & 22.83 & \hspace{1em} & - & - & - & 36.82 & \hspace{1em} & 378 & 105 & 46.12\\
NO$_{2}$ & 28 & 20549 & \hspace{1em} & 311 & 66.07 & 22.93 & \hspace{1em} & - & - & - & 40.69 & \hspace{1em} & 378 & 109 & 46.25\\
Cl$_{2}$ & 34 & 34334 & \hspace{1em} & 378 & 90.83 & 28.09 & \hspace{1em} & - & - & - & 26.58 & \hspace{1em} & 561 & 156 & 73.24\\
NaCl & 34 & 42826 & \hspace{1em} & 498 & 86.00 & 28.54 & \hspace{1em} & - & - & - & 20.46 & \hspace{1em} & 561 & 166 & 74.34\\
SF$_{2}$ & 36 & 56025 & \hspace{1em} & 567 & 98.81 & 31.96 & \hspace{1em} & - & - & - & 30.65 & \hspace{1em} & 630 & 180 & 78.67\\
HBr & 36 & 62589 & \hspace{1em} & 602 & 103.97 & 31.88 & \hspace{1em} & - & - & - & 16.03 & \hspace{1em} & 630 & 154 & 78.98\\
SO$_{2}$ & 36 & 75315 & \hspace{1em} & 691 & 108.99 & 33.21 & \hspace{1em} & - & - & - & 29.75 & \hspace{1em} & 630 & 187 & 66.90\\
NO$_{3}^{-}$ & 38 & 61132 & \hspace{1em} & 622 & 92.28 & 31.23 & \hspace{1em} & - & - & - & 65.01 & \hspace{1em} & 703 & 182 & 86.40\\
H$_{2}$Se & 38 & 69684 & \hspace{1em} & 631 & 110.43 & 33.91 & \hspace{1em} & - & - & - & 16.49 & \hspace{1em} & 703 & 196 & 86.88\\
\toprule
\end{tabular}}
\caption{{The full set of results of the numerical simulations discussed in the main text. The molecular geometry is approximately that of the equilibrium configuration. A number of results related to the collecting of Hamiltonian terms, how the arrangement reduces the number of measurements required, and the number of two-qubit gates in the resulting rotation circuits are shown. We used OpenFermion~\cite{openfermion} and Psi4~\cite{parrish_psi4_2017} to obtain Hamiltonians in the STO-3G  basis and under the symmetry conserving Bravyi-Kitaev transformation~\cite{bravyi_kitaev2002, bksc_fermion}. Using our collecting method, \textsc{Sorted Insertion}, the number of collections, $N$, the average number of terms per collection, $\overline{m_{i}}$, and the average rank of the collections, $\overline{k_{i}}$ are shown, for molecules with $n$ qubits and $t$ Pauli operators, excluding the identity, in the Hamiltonian sum.
For the smallest nine molecules, we calculated the resulting ratio $R$, as given in Eq.~\eqref{eq:rratio}, for 100 randomly selected trial states, prepared by choosing random sets of parameters for a hardware efficient ansatz preparation circuit of depth 1. For each molecule, we show the mean, minimum and maximum values of $R$ obtained from the 100 runs. We also show the value of $\hat{R}$, given by Eq.~\eqref{eq:rhatratio}, obtained, which is close to the mean value of $R$ where this has been calculated.
A key result of interest is the maximum number of two-qubit gates required to obtain a measurement of all the operators in a given Hamiltonian. We show the theoretical maximum, given by the largest value of $kn -\frac{1}{2}k(k+1)$ for any collection, and the true largest value for any collection. The ratio of these two numbers is shown in Fig.~\ref{fig:plots}(c). We also show the mean number of two-qubit gates required in a rotation circuit, averaged across all collections for a given molecule.}}
\label{tab:fulldat}
\end{table}
\end{document}